\documentclass{aims}
\usepackage{amsfonts,amsmath,amssymb}
\usepackage{paralist}
\usepackage{graphics} %% add this and next lines if pictures should be in eps format
\usepackage{epsfig} %For pictures: screened artwork should be set up with an 85 or 100 line screen
\usepackage{graphicx}  
\usepackage{epstopdf}%This is to transfer .eps figure to .pdf figure; please compile your paper using PDFLeTex or PDFTeXify.
\usepackage[colorlinks=true]{hyperref}
\hypersetup{urlcolor=blue, citecolor=red}
\usepackage{appendix}

\usepackage{float}
\usepackage{subfigure}
\usepackage{calc}
\usepackage{chngcntr}
\counterwithout{figure}{section}
\usepackage{amsrefs}

\usepackage{dcolumn} %for columns aligned at decimal pts
\newcolumntype{d}[1]{D{.}{.}{#1} }

\def\currentvolume{X}
 \def\currentissue{X}
  \def\currentyear{200X}
   \def\currentmonth{XX}
    \def\ppages{X--XX}
     \def\DOI{10.3934/xx.xx.xx.xx}

\usepackage{xcolor}

\newtheorem{theorem}{Theorem}[section]

\newtheorem{lemma}[theorem]{Lemma}

\theoremstyle{definition}

\newcommand{\ol}[1]{\overline{#1}}

\newcommand{\ed}[1]{{#1}}

\DeclareMathOperator*{\rank}{rank}

\title[Bistable Dynamics in Early Stage HIV Infection] %Use the shortened version of the full title
      {Bistable Dynamics and Hopf Bifurcation in a Refined Model of Early Stage HIV Infection}

\author[Stephen Pankavich, Nathan Neri, and Deborah Shutt]{}


 \email{pankavic@mines.edu}
 \email{nneri@mines.edu}
 \email{shuttda@vmi.edu}

\thanks{The first author is supported by NSF grants DMS-1211667, DMS-1551229, and DMS-1614586}

\thanks{$^*$ Corresponding author: Stephen Pankavich}

\begin{document}
\maketitle

\renewcommand{\thefootnote}{\arabic{footnote}}

% Enter the first author's name and address:
\centerline{\scshape Stephen Pankavich$^*$ and Nathan Neri}
\medskip
{\footnotesize
% please put the address of the first author
 \centerline{Colorado School of Mines}
   \centerline{1500 Illinois St.}
   \centerline{ Golden, CO 80401, USA}
} % Do not forget to end the {\footnotesize by the sign }

\medskip

%\centerline{\scshape Nathan Neri}
%\medskip
%{\footnotesize
% % please put the address of the second  and third author
% \centerline{Colorado School of Mines}
%   \centerline{1500 Illinois St.}
%   \centerline{ Golden, CO 80401, USA}
%}   
%   \medskip
   
\centerline{\scshape Deborah Shutt}
\medskip
{\footnotesize
 % please put the address of the second  and third author
 \centerline{Virginia Military Institute}
   \centerline{319 Letcher Ave.}
   \centerline{ Lexington, VA 24450, USA}   
}

\bigskip

% The name of the associate editor will be entered by an editorial staff
% "Communicated by the associate editor name" is not needed for special issue.
 \centerline{(Communicated by the associate editor name)}

\begin{abstract}
Recent clinical studies have shown that HIV disease pathogenesis can depend strongly on many factors at the time of transmission, including the strength of the initial viral load and the local availability of CD4+ T-cells.  In this article, a new within-host model of HIV infection that incorporates the homeostatic proliferation of T-cells is formulated and analyzed.  Due to the effects of this biological process, the influence of initial conditions on the proliferation of HIV infection is further elucidated. The identifiability of parameters within the model is investigated and a local stability analysis, which displays additional complexity in comparison to previous models, is conducted. The current study extends previous theoretical and computational work on the early stages of the disease and leads to interesting nonlinear dynamics, including a parameter region featuring bistability of infectious and viral clearance equilibria and the appearance of a Hopf bifurcation within biologically relevant parameter regimes.
\end{abstract}

\section{Introduction}
Mathematical modeling of the in-host behavior of viral infections has become an indispensable tool to biological researchers in recent decades. New models have been used to describe the dynamical behavior of various infectious diseases such as HIV, HBV, and influenza, among others. 
Within this field, testing specific hypotheses based on clinical data is often difficult since samples cannot be taken frequently from patients, and viral load detection techniques may lack a necessary level of precision.
Thus, new predictive models play a central role and are continually needed to further our understanding of disease dynamics.
One such mathematical model that has been quite useful, known as the standard model of viral dynamics \cite{NM,Alan,PP}, describes the early stage in-host behavior of HIV infection. 
In general, the time course of this disease typically consists of three distinct phases.  The first, known as the acute stage, is characterized by a rapid fluctuation in both the healthy %CD4+ 
T-cell and virion populations, usually lasting from 2 to 10 weeks \cite{Hern}.  Symptoms during this stage, often described as ``flu-like'', include fever, swollen glands, sore throat, rash, and fatigue.  During the second stage, known as chronic infection, the size of the uninfected T-cell population and viral load maintain relatively constant states, with the latter known as the viral set point.  Without the aid of antiretroviral treatment, this period can persist for up to 10 years but can vary greatly among individual patients \cite{Fauci,Shutt,Pankavich}.  Finally, within the third stage the viral load experiences exponential growth with a correspondingly rapid decrease in the healthy T-cell population.  This leads to the onset of AIDS, defined clinically as a T-cell count of an HIV-positive patient measured below $200$ cells/mm$^3$.

%We will see that the terms included in this reduced system will create more complex dynamics than previous models but dynamics which we are able to analysis and capture observed biological phenomena.

While the standard model of viral dynamics has been extremely successful in reproducing the acute and chronic stages, it has been shown that both the infected and infection-free equilibrium states of healthy T-cells, infected T-cells, and virions induced by this model are globally asymptotically stable.
This property was first investigated analytically in \cite{DS}, and later proved using a Lyapunov function in \cite{Korob}.
The global stability of these states implies that the long time asymptotic behavior of the system depends only upon parameter values in the model.  As a direct consequence, both the equilibrium values and stability properties of these equilibria are independent of initial conditions. 
However, a number of recent clinical studies \cite{Igarashi, Igarashi2, Endo, Liu} have shown that additional factors beyond these parameters %(e.g., growth, decay, infection, and viral production rates) 
may have a significant impact on the development or clearance of a persistent infection.
Such factors include the initial viral load and the availability of target CD4 T-cells at the time of transmission.
Because early events during infection may determine both the pathogenic consequences of the virus and its sensitivity to interventions or treatment strategies to combat the disease, Igarashi et al.~\cite{Igarashi} evaluated the effects of inoculum size on the development of the disease in macaques infected with Simian/Human Immunodeficiency Virus (SHIV).
In particular, the results of this study showed that macaques who were administered large intravenous SHIV inocula experienced irreversible CD4+ T lymphocyte depletion and developed clinical disease. 
In contrast, rhesus monkeys receiving $50\%$ tissue culture infective doses (or less) of virus survived the acute stage with reduced but stable levels of CD4+ T lymphocytes and produced antibodies capable of neutralizing SHIV.  
In short, although SHIV induced an extremely rapid and profound depletion of T-cells in all infected rhesus monkeys, the loss of this T-cell subset was not irreversible in animals inoculated with small amounts of virus.  A similar investigation has since been conducted in \cite{Liu} yielding analogous results.

A different study \cite{Igarashi2} has highlighted the importance of the T-cell count at the time of primary viral infection. In particular, during further studies of seventeen rhesus macaques, SHIV infection was found to emerge only in a single monkey whose T-cell count had been markedly depleted by monoclonal antibody (mAb) treatment at the time of primary viral infection, while none of the remaining sixteen monkeys inoculated with SHIV, but not treated with the mAb, developed immunodeficiency. A similar outcome was observed in \cite{Endo}. Hence, the availability of target T-cells at the time of viral transmission can also play a large role in the establishment of a persistent infection, as differing strengths of the susceptible T-cell population may promote or inhibit viral replication.  

%Since the standard model of the acute stage possesses two distinct equilibria (viral clearance and viral persistence) that are globally stable in mutually exclusive parameter
Since the equilibria of the standard model of viral dynamics are globally stable in mutually exclusive parameter regimes, these empirical results suggest that, in order to appropriately describe early HIV (or SHIV) dynamics, additional factors must be considered in the model development. Other authors \cite{Ganusov, Bonhoeffer, Elaiw, Adams} have further posited such considerations, asserting the need for a variety of secondary biological characteristics including variability of host susceptibility to initial infection, within-host competition between different viruses for target cells at the initial site of virus replication, and the effects of the innate immune response. These ingredients should play a realistic role in disease pathogenesis and long time dynamics.
In general, previous in-host models do not account for effects arising from the strength of the initial viral load or variations in the T-cell count at the time of transmission, as they describe the tendency to viral infection or clearance based solely upon parameters and not on initial conditions. 

Another element of disease pathogenesis overlooked within the standard model is the homeostatic mechanism that regulates the peripheral T-cell pool.
Recent clinical studies have displayed the importance that homeostasis of the susceptible T-cell population may play during infection \cite{Moreno, Catalfamo}, as the replenishment of target cells provides additional opportunities for HIV infection by freely moving virions.
In the current study, we focus on the homeostatic proliferation of T-cells, i.e. the process by which T-cells in a lymphopenic host divide in the absence of cognate antigen to reconstitute the peripheral lymphoid compartment, which is believed to be driven by the presence of foreign antigens \cite{Mackall, Tanchot}.
A few long-term models of HIV infection \cite{Hadji, Hern, PankLoudon} have incorporated the homeostatic proliferation of the T-cell population within their formulation, but the dynamical effects of this biological mechanism are not well-understood. 
%within the context of in-host models.
Other authors \cite{Frenchie} have considered an acute stage model incorporating a logistic growth term, depending only upon healthy T-cells, to represent the body's propensity to regulate the T-cell population.  However, in the setting of HIV-induced lymphopenia, it was determined \cite{Catalfamo} that the homeostatic proliferation of CD4+ T-cells is driven primarily in response to the viral load, while naive CD4+ T-cells are also recruited into the proliferating pool due to CD4+ T-cell depletion.  Therefore, in the presence of HIV, such a regulatory mechanism should depend on the strength of the viral load in addition to the size of the T-cell pool.

Based on the aforementioned experimental findings concerning the influence of initial conditions and T-cell homeostasis, we explore a refined model of early stage infection dynamics that incorporates the ability of the immune system to maintain the T-cell count even when the number of such cells is depleted by the presence of the virus.  In accounting for such effects, it will be shown that this model will accurately portray the dependence of equilibria on initial conditions by producing a biologically relevant parameter regime featuring bistability of the infected and uninfected equilibrium states.
Hence, the model proposed herein will account for both of the aforementioned processes.
%
%The paper proceeds as follows. 
In the next section, the new model of early HIV infection is discussed, and a study of parameter identifiability is conducted.  
In Section $3$, we prove that exactly three states exist - one uninfected and two infected equilibria. In Section $4$, the local stability properties of equilibria are characterized in terms of parameter values.  In particular, we identify a biologically-important region of the parameter space within which both the relevant infected equilibrium and the viral clearance state are locally stable.  This illustrates that the development of a persistent infection will depend crucially on initial conditions, and we further explore the basins of attraction generated by these equilibria.
Finally, we show that the system experiences a Hopf bifurcation that gives rise to oscillatory behavior within a certain parameter regime. To conclude the paper, appendices containing proofs of the aforementioned results are provided.

\section{Model and Parameters}   
The proposed dynamical model couples a nonlinear system of three ordinary differential equations given by
\begin{equation}
\tag{3CM}
\label{Acute}
	\left.
	\begin{aligned}
		\frac{dT}{dt}&=\lambda + \frac{\rho}{C + V}TV - k T V -d_T T \\
		\frac{dI}{dt}&= k T V - d_I I\\
		\frac{dV}{dt}&=p I- d_V V.
	\end{aligned}
	\right \}
\end{equation}
Here, $T(t)$ denotes the population of healthy %CD4+ 
T-cells, $I(t)$ the population of these cells which have been infected, and $V(t)$ the size of the virion population.  The parameter $\lambda$ represents the source of new cells arising from general production, while the healthy cell death rate is denoted by $d_T$.  The interaction, or mass action, term $kTV$, where $k$ is the infection rate, represents the infection of healthy T-cells and the subsequent conversion of these cells to infected lymphocytes, with corresponding death rate $d_I$.  The parameter $p$ is the rate at which new virions are created by the infected cell population, and the clearance rate of free virus particles is given by $d_V$. See Table \ref{List} for a complete list of parameters and variables with representative initial values. 
Additionally, in this model we do not consider distinct compartments within the host since the dynamics of interest take place over many weeks, while transfer between these compartments occurs on the time scale of hours.
 
\begin{table}[t]
\centering
\vspace{0.1in}
\begin{footnotesize}
\begin{tabular}{p{.2cm} p {5.5cm} p {3.5cm} p {4cm}}
\hline
& Quantity & Values / Initial Values  & References\\
\hline
\multicolumn{4}{l}{Original Populations}\\
$T$ & Uninfected CD4$^+$ T-cells & $1000~\mathrm{mm}^{-3}$ &  \cite{Hern}\\
$I$ & Infected CD4$^+$ T-cells & $0~\mathrm{mm}^{-3}$ & \cite{Hern}\\
$V$ & Wild-type HIV virions & $10^{-2}~\mathrm{mm}^{-3}$ & \cite{Hern}\\ %$10^{-2}

\\
\multicolumn{3}{l}{Dimensionless Populations ($^*$ omitted in exposition)}\\
\multicolumn{2}{l}{$T^* = \frac{pk}{d_I d_V} T$}  & $T^*_0 = 1.81$ &  \\
\multicolumn{2}{l}{$I^* = \frac{pk}{d_T d_V} I$} & $I^*_0 = 0$ & \\
\multicolumn{2}{l}{$V^* = \frac{k}{d_T} V$} & $V^*_0 = 4.57\times 10^{-5}$ & \\ %$4.78\times 10^{-3}$
\multicolumn{2}{l}{$t^* = d_T t$} &  & \\

\\
\multicolumn{3}{l}{Original Parameters}\\
$\lambda$ & Rate of supply of T-cells &$10~\mathrm{mm}^{-3}~\mathrm{day}^{-1}$ & \cite{KirschAlan} \\
$\rho$ & Maximum homeostatic growth rate & $0.01~\mathrm{day}^{-1}$ & \cite{Hern}\\
$C$ & Homeostatic half-velocity & $300~\mathrm{copies}~\mathrm{mm}^{-3}$ & \cite{Hern}\\
$k$ & Infection rate &  $4.57\times10^{-5}~\mathrm{mm}^{3}~\mathrm{day}^{-1}$ & \cite{Hern, Hadji}\\
$d_T$ & Death rate of uninfected T-cells & $0.01~\mathrm{day}^{-1}$ & \cite{Hadji, Hern}\\
$d_I$ & Death rate of infected T-cells & $0.40~\mathrm{day}^{-1}$ & \cite{KirschAlan, Hern} \\
$p$ & Rate of viral production & $38~\mathrm{virions}~\mathrm{per~cell}~\mathrm{day}^{-1}$ & \cite{Hern, Hadji}\\
$d_V$ & Clearance rate of free virus & $2.4~\mathrm{day}^{-1}$ & \cite{Hern, Hadji}\\

\\
\multicolumn{3}{l}{Dimensionless Parameters}\\
\multicolumn{2}{l}{$R_0 = \frac{\lambda k p}{d_T d_I d_V}$} & $1.81$ &  \\
\multicolumn{2}{l}{$R_m = \frac{\rho}{Ck}$} & $0.73$ & \\
\multicolumn{2}{l}{$\alpha_1 = \frac{d_I}{d_T}$} & $40$ & \\
\multicolumn{2}{l}{$\alpha_2 = \frac{d_V}{d_T}$} &  $240$ & \\
\multicolumn{2}{l}{$\beta = \frac{d_T}{Ck}$} & $0.73$ & \\
\\
\hline\\
\end{tabular}
\end{footnotesize}
\caption{\footnotesize{Variables and Parameters}
}
\label{List}
\vspace{-0.1in}
\end{table}

The term $\frac{\rho}{C + V}TV$ describes the homeostatic production of T-cells due to the presence of the virus and subsequent decline in healthy T-cells, both of which may vary over the course of infection.
Here, $\rho$ is the maximum growth rate and $C$ is the half-velocity constant of growth.
Note that the behavior of this term is limited by the growth and decay of the virus population.  
In particular, the function $M(V) = \frac{\rho V}{C + V}$ satisfies $M(0) = 0$, $M'(V) > 0$ and $\displaystyle \lim_{V \to \infty} M(V) = \rho$.  Hence, when no virions are present in the system, this so-called Michaelis-Menten term vanishes and the basic dynamics are the same as the standard virus model.
%, which can be completely recovered in the case $\rho =0$ as well.  
This is consistent with the actual immune response as the body need not further augment the T-cell population in the absence of virions. Contrastingly, as the virus population grows large, the infected host's immune system replenishes the T-cell population so as to balance the effects resulting from its depletion, and this occurs at a growing rate whose maximal impact is $\rho T$.
Regardless of the limited rate of growth within this term, the inclusion of homeostatic proliferation, as we will show, has a profound affect on the dynamics of the system. 

While the new model \eqref{Acute} can be derived from a bottom-up approach merely by adding the homeostatic proliferation term to the standard model of viral dynamics, it also stems directly from a top-down approach. More specifically, \eqref{Acute} can be fully derived from a reduced description of long-term models that were proposed in \cite{Hadji, Hern} to accurately represent all three stages of HIV infection within a host. In particular, a dynamic active subspace decomposition of the twenty-seven dimensional parameter space within the three-stage model of \cite{Hadji}, which features seven different in-host populations, was performed in \cite{PankLoudon}. This decomposition produces a global sensitivity analysis of the parameter space and indicates exactly which parameters are important to the evolution of the model during each of the three distinct phases of disease progression.  \ed{Upon eliminating those parameters (of which there were nineteen) that are found to be negligible throughout the acute stage, a total of four populations - namely the influence of latently-infected T-cells, macrophages, infected macrophages, and the cytotoxic lymphocyte response - completely decouple from the model. Hence,} the reduced system \eqref{Acute} results, providing a more precise description \ed{of the early stage behavior of the disease} than the standard viral dynamics model.
Figure~\ref{repsim} contains a representative simulation of \eqref{Acute} and includes a comparison to the early stage behavior of the long-term model of \cite{Hadji}.
We note that other models of HIV infection \cite{Banks, SmithBanks, Hern} have also incorporated such a Michaelis-Menten term to describe homeostasis, though the current article will contain the first dynamical analysis of such a model.

%Parameters \& Sensitivity Analysis (see ``Mathematical model of multivalent virus antibody complex formation in humans following acute and chronic HIV infections'') 

%\subsection{Parameter Values}
%As in other studies \cite{Hadji, Hern}, the unknown parameters are fit to previously-obtained patient data \cite{Fauci}.
%In this case, all parameters but $d_T$ (whose value has been well-established) were fit to the patient data with initial values provided by the references displayed in Table \ref{List}, and this was performed using Matlab's constrained nonlinear solver, which utilizes the Nelder-Mead Simplex method.
%In particular, the new parameters $\rho$ and $C$ were adjusted to obtain appropriate HIV trajectories with respect to clinical observations in \cite{Fauci}. %Greenough et al. 1999
%Figure~\ref{repsim} contains a representative simulation of the model \eqref{Acute} with fitted parameter values and includes a comparison to the acute stage behavior of the full, three-stage model of \cite{Hadji}, which is much more complex and features seven components and twenty-seven parameters.

\subsection{Dimensionless system}
To reduce the size of the parameter space, the original model \eqref{Acute} is recast in dimensionless form.  The resulting system, in which dimensionless populations have been renamed $T^*$, $I^*$, and $V^*$ is
\begin{equation}
\tag{3CM*}
\label{AcuteNd}
	\left.
	\begin{aligned}
		\frac{dT^*}{dt}&=R_0 + \frac{R_m}{1 + \beta V^*}T^*V^* - T^* V^* - T^* \\
		\frac{dI^*}{dt}&= \alpha_1 \left ( T^* V^* - I^* \right )\\
		\frac{dV^*}{dt}&= \alpha_2 (I^* - V^*).
	\end{aligned}
	\right \}
\end{equation}
where
\begin{equation}
\label{R0Rm}
R_0 = \frac{\lambda k p}{d_T d_I d_V}, \quad 
R_m = \frac{\rho}{Ck}, \quad
\alpha_1 = \frac{d_I}{d_T}, \quad
\alpha_2 = \frac{d_V}{d_T}, \quad
\beta = \frac{d_T}{Ck}.
\end{equation}
The values of dimensionless parameters are summarized within Table \ref{List}, and in the future we will remove the $^*$ notation and deal solely with the dimensionless system.
Notice that each new parameter is positive since the original variables are positive.
The complete derivation of \eqref{AcuteNd} from \eqref{Acute} can be found in Appendix~\ref{appA}.
The model \eqref{AcuteNd} contains only five parameters, and each may play a role in the dynamics of the system.  However, we will typically fix the values of $\alpha_1, \alpha_2$, and $\beta$ while considering variations in $R_0$ and $R_m$, which represent the usual basic reproduction number (as in the standard viral model) and a new reproduction number generated by the addition of the Michaelis-Menten term, respectively.

\begin{figure}[t]
%\subfigure[Real part of eigenvalues $\eta_2$ and $\eta_3$]{\includegraphics[width=7.5cm]{Eiplusrealpart.eps}}
\hspace{-0.4in}
\subfigure{\includegraphics[width=7cm]{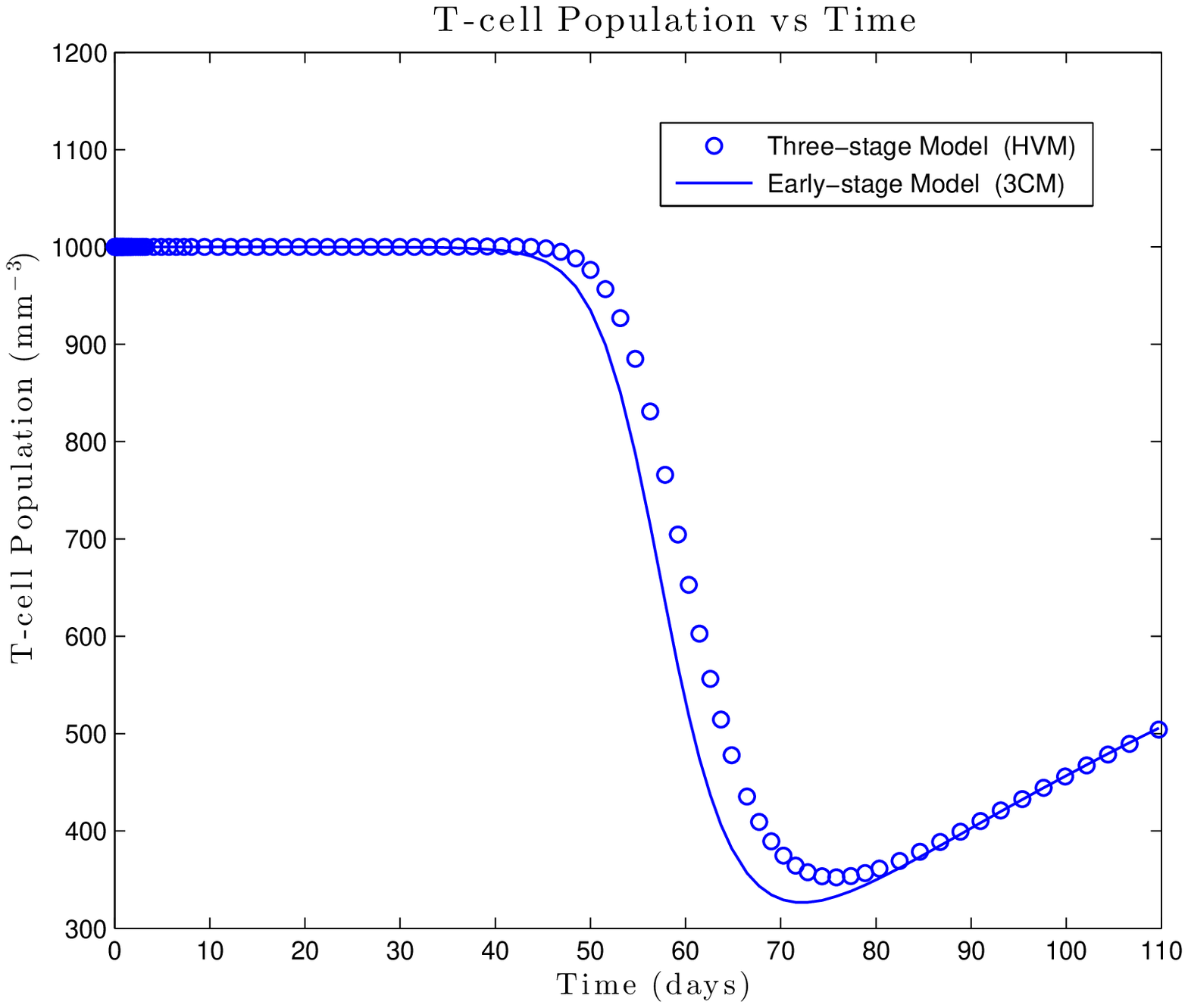}}
\hspace{-0.25in}
\subfigure{\includegraphics[width=7cm]{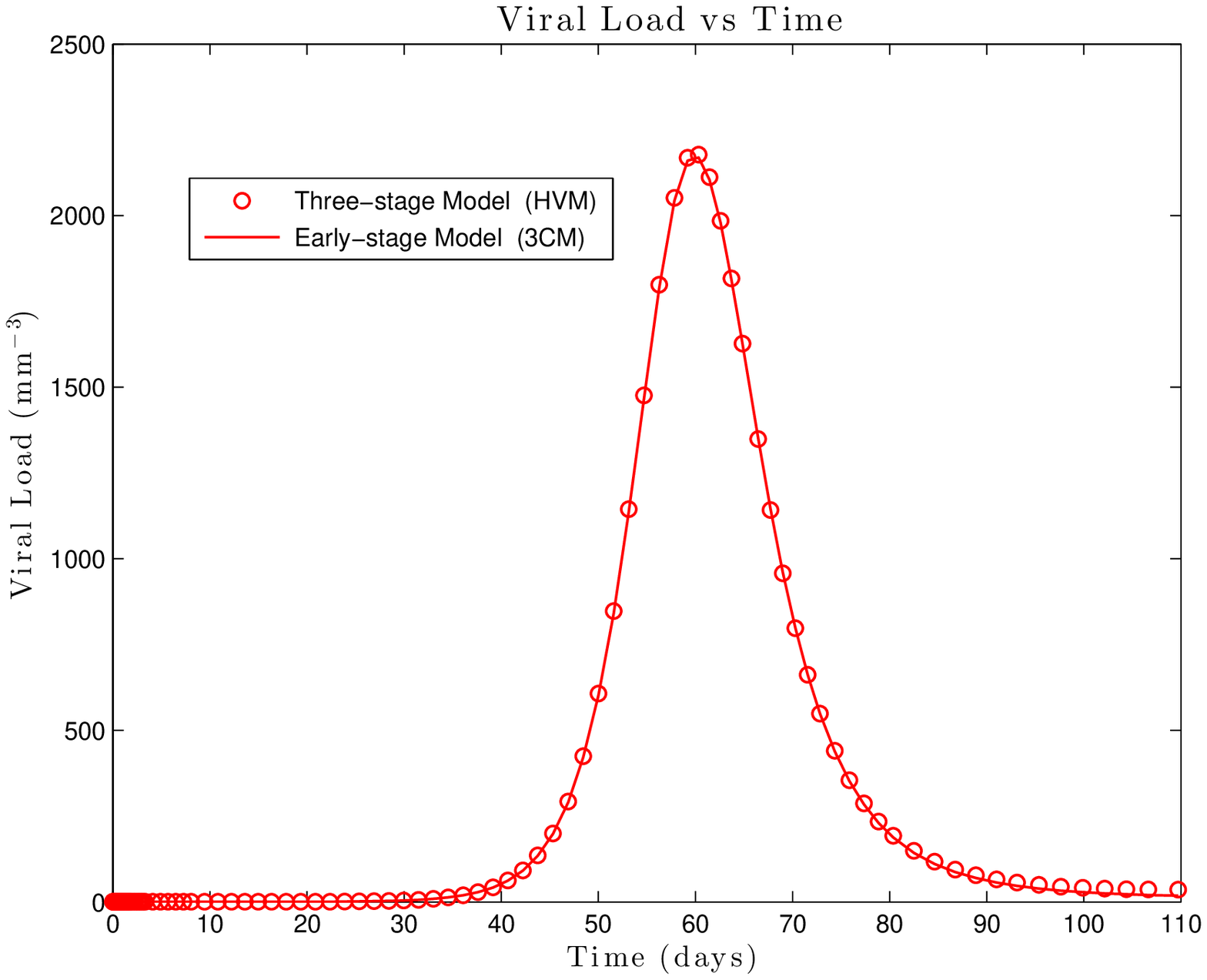}}
\vspace{-0.2in}
\caption{\footnotesize{A representative simulation of \eqref{Acute} with parameter values given in Table \ref{List}, and a comparison with the full three-stage model in \cite{Hadji} - T-cell count (left) and viral load (right).}
}
\label{repsim}
\end{figure}

\subsection{Parameter Identifiability}
With the dimensionless system determined, we study parameter identifiability in \eqref{AcuteNd} as this model can provide useful simulations only if the parameters involved can be discerned from data.  
In particular, we first conduct a test of \eqref{AcuteNd} developed for differing models in \cite{Miao}, \cite{Wu}, and \cite{XiaMoog} to understand the \ed{\emph{structural}} identifiability of parameters. 
\ed{Structural identifiability is used to characterize the one-to-one property of the map that takes the parameter space to the set of system outputs (i.e., the information encapsulated by collected data). In order to evaluate this property for \eqref{AcuteNd} we use the so-called Multiple Time Points (MTP) method developed in \cite{Wu}, which entails the construction of an invertible identification function $\Phi(\theta)$, from the parameter space to the set of observable outputs, that preserves the structure of the differential equations model. In particular, because invertibility of such a mapping is required, we wish to ultimately conclude that $\frac{\partial \Phi}{\partial \theta}$ has full rank.} 

\ed{To begin, we first describe the space of output values. Because healthy and infected T-cell counts can be both difficult to measure and unreliable, data is most easily gathered from an individual's viral load. Hence, model outputs in this context will be regarded as values of the viral load and its derivatives, as the latter are needed to compensate for the lack of T-cell data but can be generated from values of $V$. Thus, we begin to construct $\Phi$ by first eliminating the populations $T$ and $I$ within \eqref{AcuteNd} in favor of derivatives of $V$. This procedure involves merely taking derivatives in \eqref{AcuteNd} and representing $T$ and $I$ in terms of $\dot{V}$, $\ddot{V}$, and $\dddot{V}$ and yields a single equation to represent the original three-dimensional system of ODEs, namely
$$\dddot{V} - f(V, \dot{V}, \ddot{V}, \theta, t) = 0,$$
where $f$ is given by \eqref{V_ttt_form} below and $\theta = (R_0, R_m, \alpha_1, \alpha_2, \beta)^T$ is the vector of parameters. Hence, any solution of \eqref{AcuteNd} can be characterized by satisfying this relationship at time $t$.}

\ed{In order to identify the five distinct parameters in the model, five identification equations are needed, and this requires us to satisfy the above ODE at five different time points, say $t_k$, for $k =1,...,5$. Given this, we denote $V_k= V(t_k)$, with the same notation for derivatives (e.g., $\dot{V}_k = \dot{V}(t_k)$), and construct the identification function $\Phi: \mathbb{R}^5 \to \mathbb{R}^5$ defined by
\begin{equation}\label{big_phi_specific_nondim}
\Phi (\theta) = \begin{bmatrix}\dddot{V}_1 - f(V_1, \dot{V}_1, \ddot{V}_1, \theta, t_1)  \\ \dddot{V}_2 - f(V_2, \dot{V}_2, \ddot{V}_2, \theta, t_2)  \\ \vdots \\ \dddot{V}_5 - f(V_5, \dot{V}_5, \ddot{V}_5, \theta, t_5)  \end{bmatrix}
\end{equation}
where}
\ed{\begin{equation}
\label{V_ttt_form}
	\begin{aligned}
	f(V, \dot{V}, \ddot{V}, \theta, t) &=\alpha_1 \alpha_2 R_0 V +\left( \frac{R_m}{1 + \beta V} - 1 \right)  \left( \ddot{V} + \alpha_1 \dot{V} + \alpha_2 \dot{V} + \alpha_1 \alpha_2 V \right)V \\
	&~~~+  \left( \dot{V} - V - (\alpha_1 + \alpha_2) V\right) \left( \ddot{V} + (\alpha_1 + \alpha_2) \dot{V} + \alpha_1 \alpha_2 V \right) \frac{1}{V} \\
	&~~~+ \alpha_1 \alpha_2 (\alpha_1 + \alpha_2) \left( \frac{1}{\alpha_2} \dot{V} + V\right) + \alpha_2^2 \dot{V}.
	\end{aligned}
\end{equation}
By construction, the model \eqref{AcuteNd} is trivially satisfied at fitted parameter values $\theta^*$ given in Table \ref{List} as $\Phi(\theta^*) = 0$, and we are thus interested in whether $\Phi$ is invertible for values of $\theta \approx \theta^*$.} 

\ed{Now, computation of $\Phi$ requires knowledge of $V$ and its derivatives at five different time points, but values for these derivatives are typically unavailable either as collected data or via direct simulation of the model. Thus, we must require additional values of the viral load in order to compute them. In particular, eight values of $V$ are needed to numerically approximate $\dddot{V_k}$ for $k=1,...,5$ by using a suitable finite difference approximation. 
Therefore, we choose three additional time values $t_6, t_7$, and $t_8$ at which $V$ must be known, and note that the values of $\dot{V}_k, \ddot{V}_k$, and $\dddot{V_k}$ for $k = 1,...,5$ are merely determined by values of the viral load at multiple time points; for example, $\dddot{V_k}$ depends upon $V_1, ..., V_8$ for every $k = 1,..., 5$.}

\ed{In order to conclude that the model is locally structurally identifiable, we must show that the corresponding Jacobian matrix $\frac{\partial \Phi}{\partial \theta}$ is invertible near the fitted values $\theta = \theta^*$.  Since this is nearly impossible to perform analytically, we instead take a computational approach. First, we symbolically represent the matrix $\frac{\partial \Phi}{\partial \theta}$ using \eqref{big_phi_specific_nondim} and \eqref{V_ttt_form}. Then, fixing a specific vector of parameter values $\theta$, we simulate the output variable $V(t;\theta)$ at a chosen sequence of times using \eqref{AcuteNd} and Matlab's \texttt{ode15s} solver, and then numerically approximate its derivatives. Finally, we use these simulated values of $V_k, \dot{V}_k$, $\ddot{V}_k$, and $\dddot{V_k}$ to compute the resulting rank of the Jacobian for these particular parameter values. Repeating this calculation over a grid of parameter values within the biologically reasonable ranges $[0.5\theta^*_k, 1.5\theta^*_k]$ for $k=1,...,5$, we find $\rank \left (\frac{\partial \Phi}{\partial \theta} \right) = 5$ for every simulation. Thus, we conclude that $\frac{\partial \Phi}{\partial \theta}$ is of full rank and the associated parameters are structurally identifiable, at least locally, in the range
of parameter values used for these simulations.}

While this analysis provides a theoretical assurance that parameter values can be identified from exactly observed viral load data, clinical measurements will always contain some level of error. Even a model such as \eqref{AcuteNd}, in which parameters can be uniquely identified locally within the parameter space, may yield unreliable parameter estimates due to noisy fluctuations or measurement error in the data. 
Hence, we also study the \ed{\emph{practical}} identifiability of parameters, namely the relative proximity of fit parameters to their true values given uncertainty within obtained data, by using a Monte Carlo method outlined within \cite{Miao} and \cite{Wu}.

To estimate the differences in parameter fits generated from noisy data, we will use a metric known as the average relative estimation error (ARE).
Prior to precisely defining this quantity, we outline the algorithm for generating such values. 
Beginning with the previously fit vector of parameters $\theta^*$, we first use the numerical ODE solver to generate a time course of baseline viral load values $V^*_j = V(t_j; \theta^*)$ for a chosen set of times $t_j$, $j = 1,...,M$. 
Next, we choose a sensitivity threshold $\delta > 0$ and a number of Monte Carlo trials $N \in \mathbb{N}$, then define the perturbed viral load data
\begin{equation}
\label{eps}
\hat V^n_j = V^*_j + \epsilon^n_j, \qquad j = 1,..., M, \quad n = 1,...,N
\end{equation}
where $\epsilon^n_j \sim N(0, \delta)$ represents an unbiased normally-distributed measurement error for each fixed $j$ and $n$. 
With random error introduced within the simulated data, we perform $N$ parameter fits of this data to generate $N$ new vectors of parameter values $\hat \theta^n$, for $n = 1,...,N$.
Finally, to evaluate the variations in these fits generated by the noise, we define the ARE for each parameter by
\begin{equation}
\label{ARE}
\text{ARE}_{\delta, k} = \frac{1}{N} \sum_{n=1}^N \frac{|\hat \theta^n_k - \theta^*_k|}{|\theta^*_k|} \cdot 100\%,
\end{equation}
where $\hat \theta^n_k$ is the estimate of the $k$th parameter of $\hat \theta^n$ arising from the $n$th perturbation with variance $\delta$.  
Hence, this Monte Carlo method simulates the introduction of Gaussian measurement noise within the viral load data \eqref{eps} based on the output model \eqref{AcuteNd} and computes the expected response in parameter values from these variations.

The ARE algorithm was applied to simulations of $N=1000$ distinct simulated noisy measurement sets for each $\delta$ noise level of $5, 10, 15, ..., 30$ percent of the fit value $\theta^*_k$, at time points $t = 0, 1, 2, ..., 90$. The results are summarized in Table \ref{ARE_table}. 
%Simulated measurements were made at parameter values $\theta^*$ in \ref{Parameters_Fitted} using Matlab's built in \texttt{ode15s} solver for stiff equations.  Ninety-five measurements were taken, corresponding to the number of measurements and times used in parameter fitting in Chapter 4.
Hence, we find a collection of small relative errors, with only the error in $\beta$ rising above $\delta\%$ of the true value.
% - likely due to the small magnitude of $\beta^*$ relative to other parameters.  
That being said, $\beta$ also displays relatively minor fluctuations throughout the simulations, even for noisy data, and other parameters possess even less deviation from their fit values. Thus, \eqref{AcuteNd} appears to be quite robust with respect to variations in measurement data for the purposes of parameter fitting.

\begin{table}
	\begin{center}
		\begin{tabular*}{\textwidth}{@{\extracolsep{\fill}} c|ccccc}
			\hline
			Noise level & \multicolumn{5}{c}{Calculated $ARE$ in $\%$ of fitted value }\\
			 $\delta$ in $\%$ & $\alpha_1$ & $\alpha_2$ & $\beta$ & $R_0$ & $R_m$ \\
			\hline 
			5 & 2.6066 & 5.1185 & 8.5814 & 3.2080 & 4.1637\\
			10 & 3.5020 & 6.6600 & 14.9092 & 4.7260 & 6.0953\\
			15 & 4.2652 & 7.4536 & 20.0601 & 6.5109 & 8.0292\\
			20 & 4.6269 & 8.5138  & 24.2512 & 7.7800 & 8.6645\\
			25 & 5.2935 & 9.6199 & 27.6901 & 5.5697 & 9.7956\\
			30 & 5.6840 & 9.9405 & 30.8474 & 9.9440 & 10.9832 \\ \hline
		\end{tabular*}
	\end{center}
\caption{\label{ARE_table} \footnotesize{Calculated $ARE$ of each parameter with $N=1000$ trials.}
}
\end{table}

\section{Steady States}

We begin an analysis of the dynamics of \eqref{AcuteNd} by first determining all steady states and investigating their regions of biological relevance within the parameter space. This information will be used extensively in the next section in which the local dynamics of solutions is characterized.
We first compute the associated steady states, which are given in the form of an ordered triple $(T,I,V)^{T}$.
In particular, we find exactly three solutions to the algebraic system guaranteed by
\begin{equation}
\label{SteadyNd}
	\left.
	\begin{aligned}
	R_0 + \frac{R_m}{1 + \beta V}TV - T V - T  &= 0\\
	\alpha_1 \left ( T V - I \right ) &= 0\\
	\alpha_2 (I- V) &= 0
	\end{aligned}
	\right \}
\end{equation}
and they are summarized within the following theorem.

%\subsection{Steady States}

%\begin{figure}[t]
%\centering
%\includegraphics[height=.5\textwidth]{../Figures/Existence_Positivity/RegionsofExistence(Rm,Ro)-plane_new.eps}
%\caption{Regions of existence (i.e. real and positive values) for the uninfected steady state, $E_u$, and infected steady states, $E_i^{+}$ and $E_i^{-}$ in the $(R_m,R_0)$ plane.}
%\label{existold}
%\end{figure}

\begin{figure}[t]
\centering
\includegraphics[height=.5\textwidth]{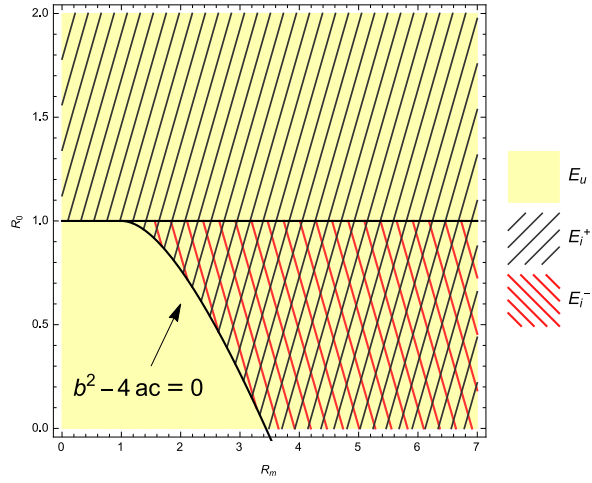}
\caption{\footnotesize{Regions of existence (i.e. real and positive values) for the uninfected state, $E_u$, and infected states, $E_i^{+}$ and $E_i^{-}$ in the $(R_m,R_0)$ plane.}
}
\label{exist}
\end{figure}

\begin{theorem}
\label{TSS}
The only time-independent solutions of \eqref{AcuteNd} are
\begin{equation*}
\begin{Large}
%\centering
\mathbf{E_u} = \left(
\def\arraystretch{2}\begin{array}{c}
	R_0\\
	0\\ 
	0
\end{array}\right)
\qquad
\mathbf{E_i^{+}} = \left(
\def\arraystretch{2}\begin{array}{c}
	1\\
	\frac{-b + \sqrt{b^{2} - 4 a c}}{2 a}\\
	\frac{-b + \sqrt{b^{2} - 4 a c}}{2 a}	
\end{array}\right)
\qquad
\mathbf{E_i^{-}} = \left(
\def\arraystretch{2}\begin{array}{c}
	1\\
	\frac{-b - \sqrt{b^{2} - 4 a c}}{2 a}\\
	\frac{-b - \sqrt{b^{2} - 4 a c}}{2 a}	
\end{array}\right)
\end{Large}
\end{equation*} 
where $a$, $b$, and $c$ are defined by
$$a = \beta, \qquad b = 1- R_m + \beta(1-R_0), \qquad c = 1- R_0$$
and the dimensionless parameters are defined by \eqref{R0Rm}.
%$$R_0 := \frac{\lambda p k}{d_I d_T d_V}, \qquad  \qquad R_m := \frac{\rho}{C k}.$$
\end{theorem}
  
Here, $E_u$ is the uninfected steady state while $E_i^{+}$ and $E_i^{-}$ represent states of persistent infection.  
In the future, when referring to components of equilibria, we will use a bar to distinguish between the components of these states (e.g., $\overline{V}$) and time-dependent solutions (e.g., $V(t)$).
Additionally, we will distinguish amongst the same components of different equilibria using subscripts (i.e., $\ol{V}_u, \ol{V}_+$, and $\ol{V}_-$).
Since all components represent scaled population sizes, we impose restrictions on the values for which these steady states are biologically reasonable.  All three populations of $E_u$ will remain nonnegative for all times.  However, for both $E_i^{+}$ and $E_i^{-}$, we must require that the infected T-cell population, $\ol{I}$, and the virus population, $\ol{V}$, be real and positive.  
To ensure real valued populations, we impose the restriction $b^{2} - 4 a c \geq 0$, which is equivalent to the condition
\begin{equation}
\tag{E$_{\mathrm{real}}$}
R_0 \geq 1 - \frac{1}{\beta}(1 - \sqrt{R_m})^{2}.
\end{equation}
The requirement that all populations of the infected states be positive forces other restrictions. 
For positivity of the $E_i^+$ state, we must impose either 
\begin{equation}
\tag{E$^+_1$}
\label{Eplus1}
R_0 > 1 + \frac{1}{\beta}(1 - R_m)
\end{equation}
or the condition
\begin{equation}
\tag{E$^+_2$}
\label{Eplus2}
1 < R_0 \leq 1 + \frac{1}{\beta}(1 - R_m).
\end{equation}
For all populations within the $E_i^-$ state to remain positive, we must impose the condition
\begin{equation}
\tag{E$^-$}
1 + \frac{1}{\beta}(1 - R_m) < R_0 < 1.
\end{equation}
These constraints are justified within Appendix~\ref{appB}. %\ref{appA}.
Additionally, Figure \ref{exist} provides a graphical summary of the restrictions on parameters necessary to guarantee positivity of corresponding equilibria.
We will often refer to such a region as the ``region of existence'' of an equilibrium state, and in studying equilibria, we will always assume that parameters are within the region of existence of the state under consideration.
Clearly, these restrictions depend upon only three parameters - $\beta$, $R_0$, and $R_m$.  However, because $\beta$ does not greatly affect the qualitative structure of the system, we will fix this parameter and focus on the behavior of the system depending only upon $R_0$ and $R_m$.  A similar approach will be taken in the investigation of stability properties of equilibria, which may further depend upon $\alpha_1$ and $\alpha_2$, but these two additional parameters are ratios of death and clearance rates, which are fairly well-known.  Thus, we will later fix these parameters as well, and again focus on the relationship between $R_0$ and $R_m$. 

%\begin{figure}[t]
%\centering
%\includegraphics[height=.5\textwidth]{../Figures/Existence_Positivity/Stability_Regions_ColorCorrected.eps}
%\caption{Regions of stability in the $(R_m,R_0)$ plane when $\alpha_1, \alpha_2, \beta$ are fixed to the values in Table~\ref{List}. } 
%\label{las-stableold}
%\end{figure}

\begin{figure}[t]
\centering
\includegraphics[height=.5\textwidth]{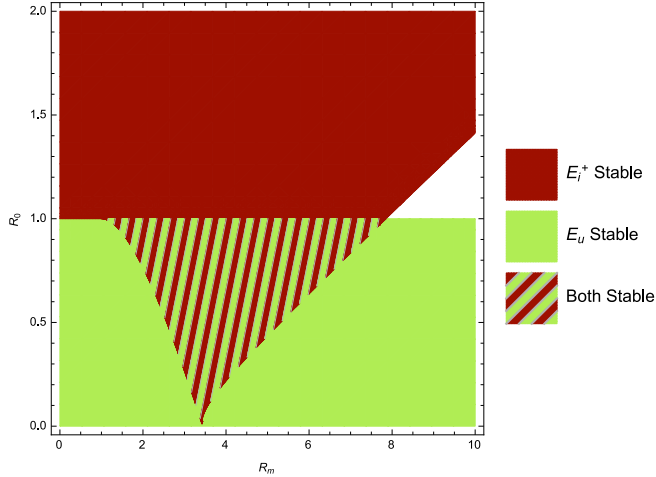}
\caption{\footnotesize{Regions of stability in the $(R_m,R_0)$ plane when $\alpha_1, \alpha_2, \beta$ are fixed to the values in Table~\ref{List}. }
} 
\label{las-stable}
\end{figure}

%\begin{table}[t]                                
%\centering    
%\vspace{0.1in} 
%\begin{footnotesize}    
%%%(T/I/V on the side)                     
%\begin{tabular*}{\textwidth}{@{\extracolsep{\fill} }cd{3}d{3}d{3}d{3}d{3}d{3}d{3}d{3}}     
%\hline                                                            
%& \lambda & d_T & d_I & d_V & k & p_T & \rho & c_T \\        
%\hline               
%$T$ & 0 & 0 & 933.8723 & 274.9009 & -136877.0143 & -1226.0796 & 0 & 0 \\   
%$I$ & 5.2492 & -3431.4745 & -76.0533 & -10.6504 & -10326.7145 & 47.5018 & 383.7150 & -1.5018 \\
%$V$ & 1.1769 & -769.3753 & -17.0520 & -5.0195 & -2315.3659 & 22.3875 & 86.0332 & -0.3367 \\            
%\hline
%\end{tabular*}    
%%%% Transposed Version (with T/I/V on top, the rest on the side)
%%\begin{tabular*}{\textwidth}{@{\extracolsep{\fill} }cd{3}d{3}d{3}}     
%%\hline                     
%% & T & I & V \\          
%%\hline      
%%$\lambda$ & 0 & 5.2492 & 1.1769 \\         
%%$d_T$ & 0 & -3431.4745 & -769.3753 \\      
%%$d_I$ & 933.8723 & -76.0533 & -17.0520 \\       
%%$d_V$ & 274.9009 & -10.6504 & -5.0195 \\        
%%$k$ & -136877.0143 & -10326.7145 & -2315.3659 \\
%%$p_T$ & -1226.0796 & 47.5018 & 22.3875 \\       
%%$\rho$ & 0 & 383.7150 & 86.0332 \\         
%%$c_T$ & 0 & -1.5018 & -0.3367 \\           
%%\hline
%%\end{tabular*}              
%\end{footnotesize}                         
%\caption{Local Sensitivity of Steady States}                           
%\label{Sensitivity List}                      
%\end{table} 

\section{Stability of Equilibria}
Having established conditions guaranteeing their biological relevance, we next examine conditions which guarantee the local stability of equilibria. To do so, we will utilize some standard dynamical tools such as the Hartman-Grobman Theorem applied to the linearization of \eqref{AcuteNd} and the Routh-Hurwitz criteria. The following result, the proof of which can be found in Appendix~\ref{appC}, %\ref{appB},
provides precise conditions on the parameter space that yield local stability and thus viral clearance or persistence.

\begin{theorem}
\label{T1}
If $R_0 < 1$ then the infection-free equilibrium $E_u$ is locally asymptotically stable, whereas if $R_0 > 1$ then it is unstable.
Additionally, for all parameter values that guarantee the positivity of the components of $E_i^-$, this equilibrium is unstable.
Finally, let $\overline{V}_+$ denote the value of the viral load for $E_i^+$, which can be expressed in terms of $\beta$, $R_0$, and $R_m$ by Theorem 3.1. Then, the equilibrium $E_i^+$ is locally asymptotically stable if the condition
\begin{equation}
\label{stabilitycondition}
%0 < 
\frac{\alpha_1 \alpha_2}{1 + \beta\overline{V}_+} \left [ \beta \overline{V}_+^2 + R_0 - 1 \right ] <  (\alpha_1 + \alpha_2)R_0 (\alpha_1 + \alpha_2 + R_0)
\end{equation}
is satisfied, and otherwise unstable. These parameter regimes are summarized by Figure \ref{las-stable}.
\end{theorem}

In short, only the uninfected steady state, $E_u$, and the infected steady state, $E_i^{+}$, are locally asymptotically stable within their respective biologically relevant regions, as described by Figure \ref{las-stable}.  
The (light) green region denotes the portion of the $(R_m,R_0)$ plane in which $E_u$ is locally asymptotically stable, while the area that is shaded (dark) red denotes the corresponding local stability region for the $E_i^{+}$ state. Interestingly, there is a small overlap of these two regions in which both steady states are locally stable, namely the striped triangular region.
%As previously mentioned, though the theorem indicates that stability properties depend on all five dimensionless parameters, we have fixed the values of $\alpha_1$, $\alpha_2$, and $\beta$ in order to concentrate on the interplay between $R_0$ and $R_m$, as the former three parameters possess only subtle variations within the literature.
%Additionally, an
An illustration of the change in the regions of Figure \ref{las-stable} generated by differing values of $\beta$ is provided in Figure \ref{changing_beta}.  
Finally, we note that with the original fitted parameter values, the reproduction numbers are $R_0 = 1.53$ and $R_m = 0.923$, respectively, which corresponds to the development of a persistent viral infection.

\begin{figure}[t]
\centering
\hspace{-0.75in}
\subfigure[$\beta = \frac{1}{2}\beta^*$]{
%\subfigure[$\beta = \frac{1}{6}\beta^*$]{
%	~~~\includegraphics[height= 6cm]{../Figures/Existence_Positivity/RegionsofExistence(Rm,Ro)-plane_new_Beta_half.eps}
%	\hspace{0.5in} \includegraphics[height = 4.5cm]{../Figures/Existence_Positivity/Stability_Regions_Beta_Half.eps}
	\hspace{0.5in} \includegraphics[height = 3.75cm]{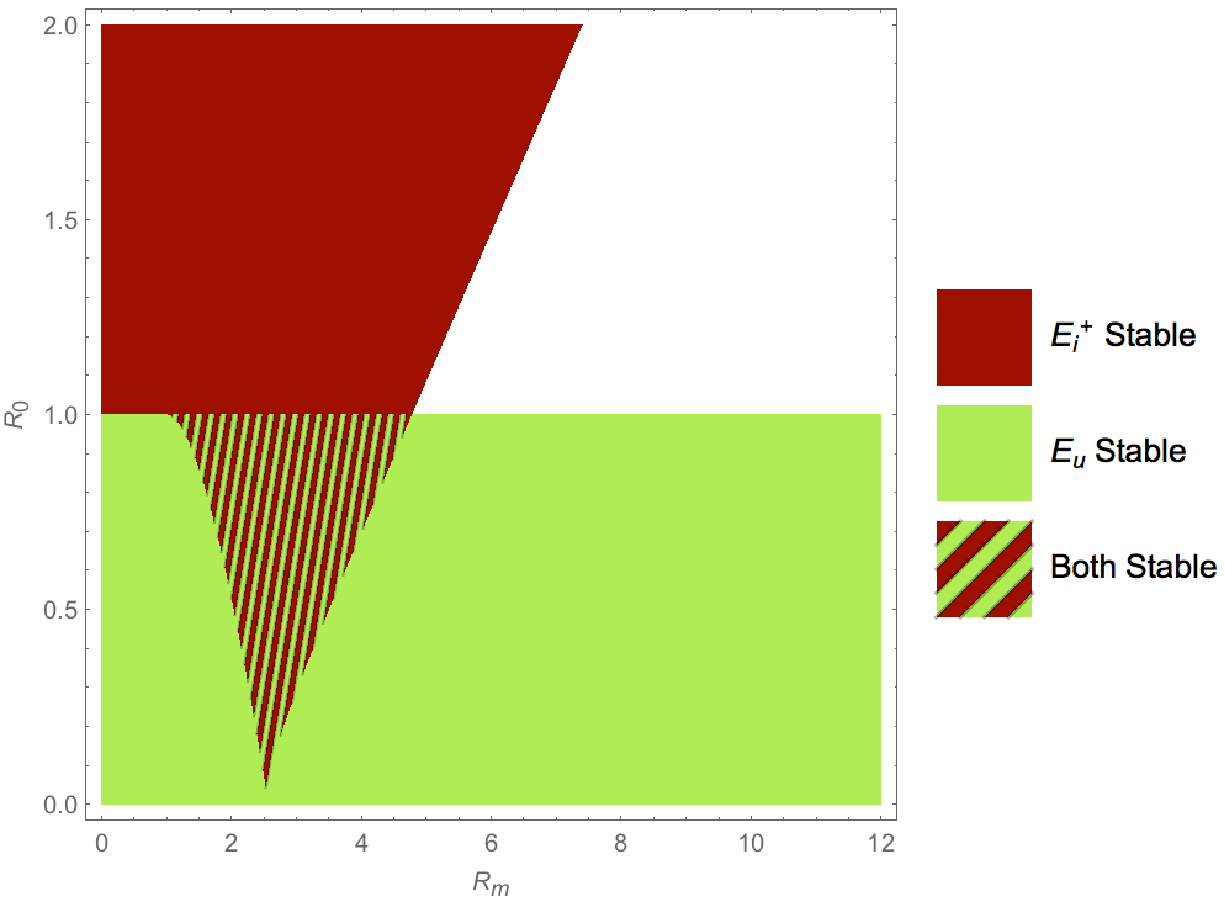}
%	\hspace{0.5in} \includegraphics[height = 3.75cm]{../Figures/Existence_Positivity/RegionsOfEquilibriumStability_newparams_3-11-18_onesixth_beta.eps}
} 
\hspace{-1.35in}
\subfigure[$\beta = \beta^*$]{
%	~~~\includegraphics[height= 6cm]{../Figures/Existence_Positivity/RegionsofExistence(Rm,Ro)-plane_new_Beta_Normal.eps}
\hspace{0.6in}	\includegraphics[height= 3.75cm]{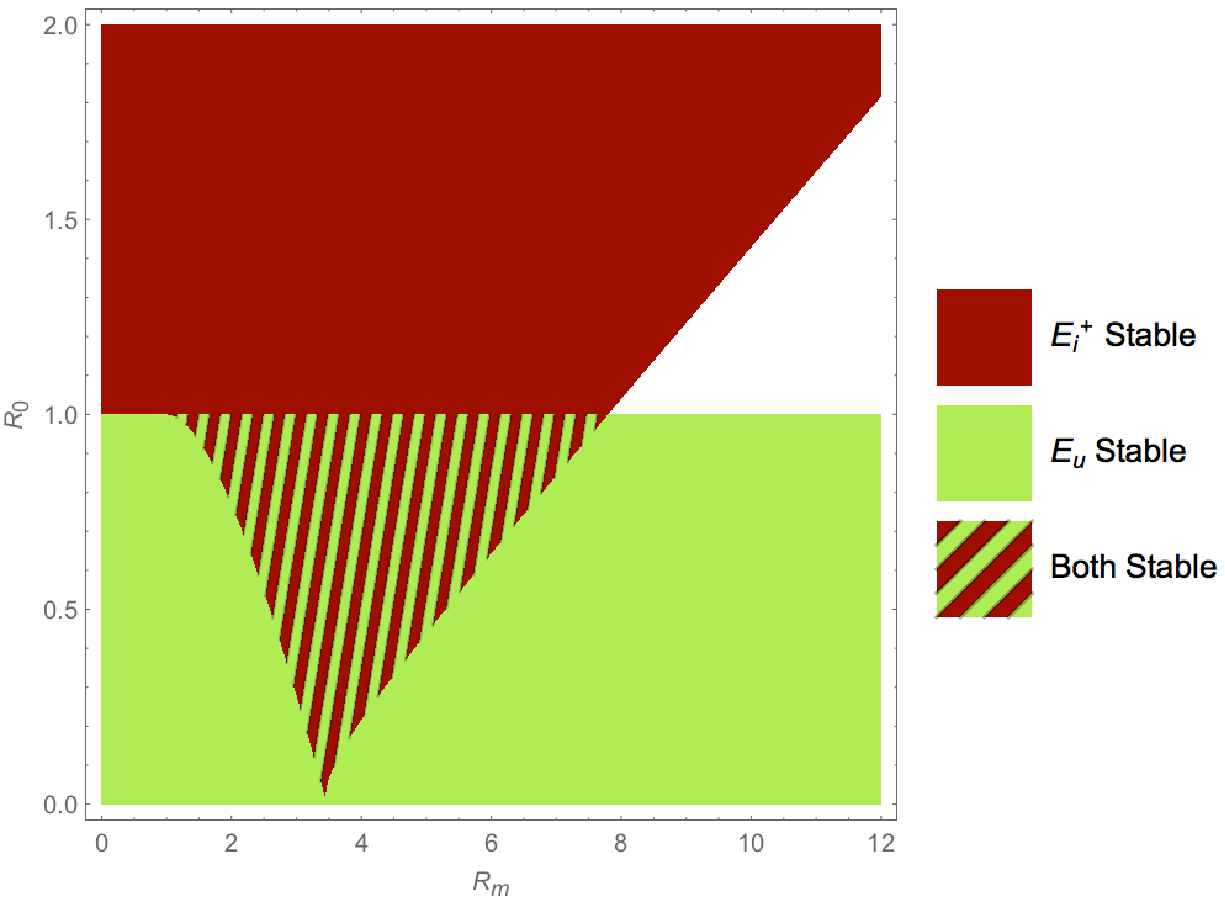}
}
\hspace{-1.3in}
\subfigure[$\beta = 2\beta^*$]{
%	~~~\includegraphics[height= 6.cm]{../Figures/Existence_Positivity/RegionsofExistence(Rm,Ro)-plane_new_Beta_Double.eps}
%\hspace{0.55in}	\includegraphics[height= 4.5cm]{../Figures/Existence_Positivity/Stability_Regions_Beta_Twice.eps}
\hspace{0.55in}	\includegraphics[height= 3.75cm]{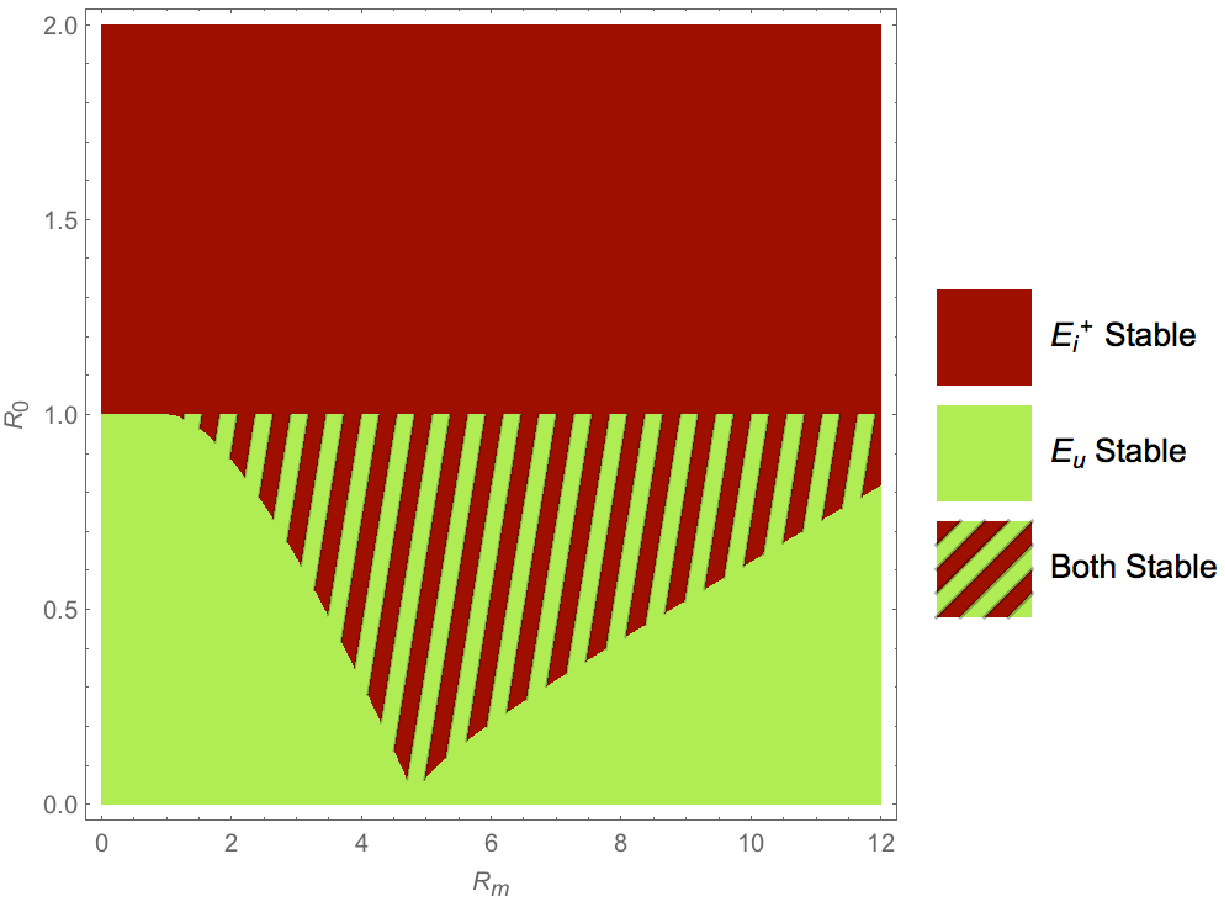}
}
\hspace{-0.3in}
\caption{\footnotesize{Changes to stability of equilibria given differing values of $\beta$ where $\beta^* =  0.73$ is the fitted value for $\beta$ in Table \ref{List}.}
}
\label{changing_beta}
\end{figure}

\begin{figure}[t]
%\subfigure[Real part of eigenvalues $\eta_2$ and $\eta_3$]{\includegraphics[width=7.5cm]{Eiplusrealpart.eps}}
\hspace{-0.3in}
\subfigure{\includegraphics[width=6.75cm]{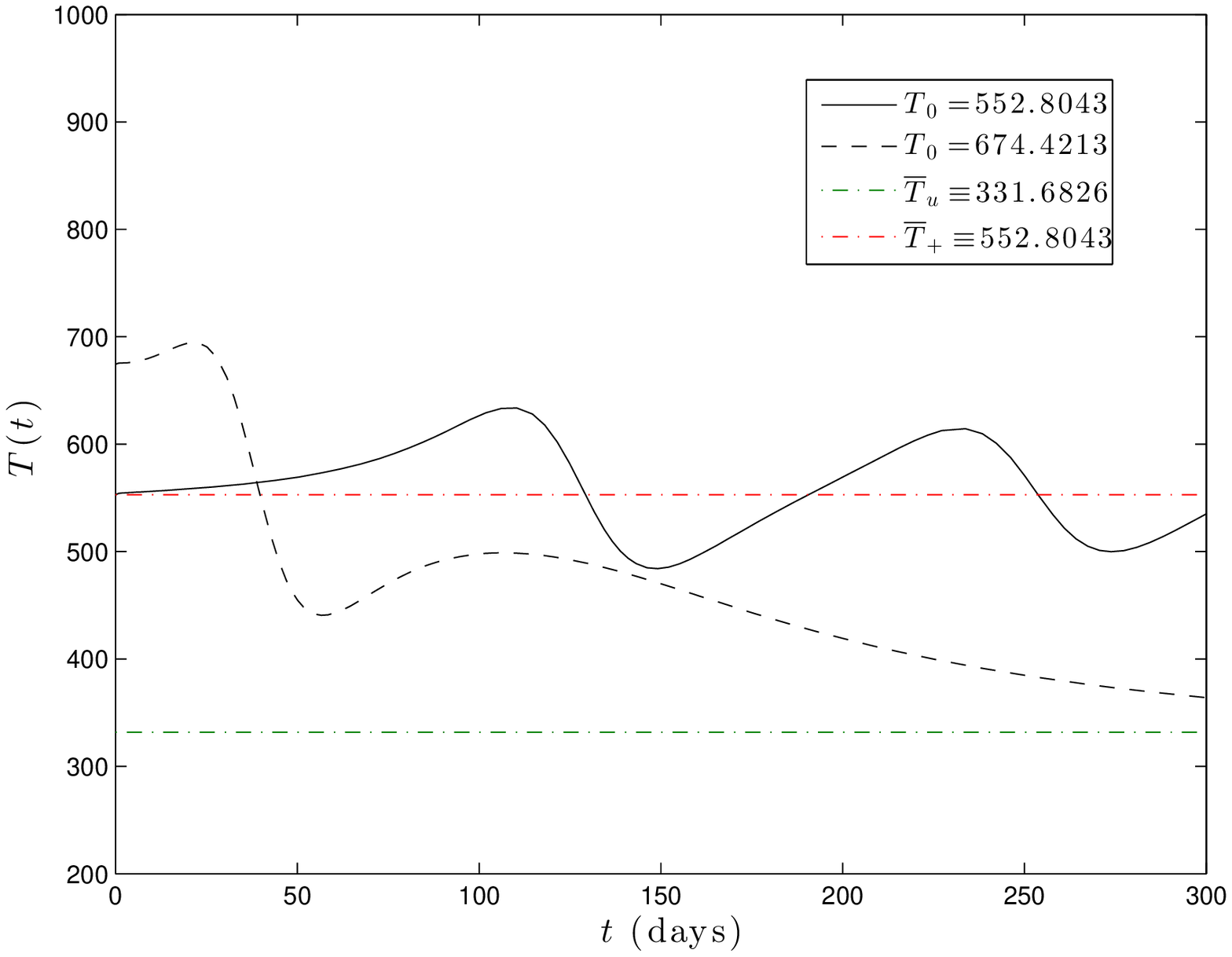}}
%\subfigure{\includegraphics[width=6.75cm]{../Figures/Stability/T0_change_T_4-29-16.eps}}
\hspace{-0.25in}
\subfigure{\includegraphics[width=6.75cm]{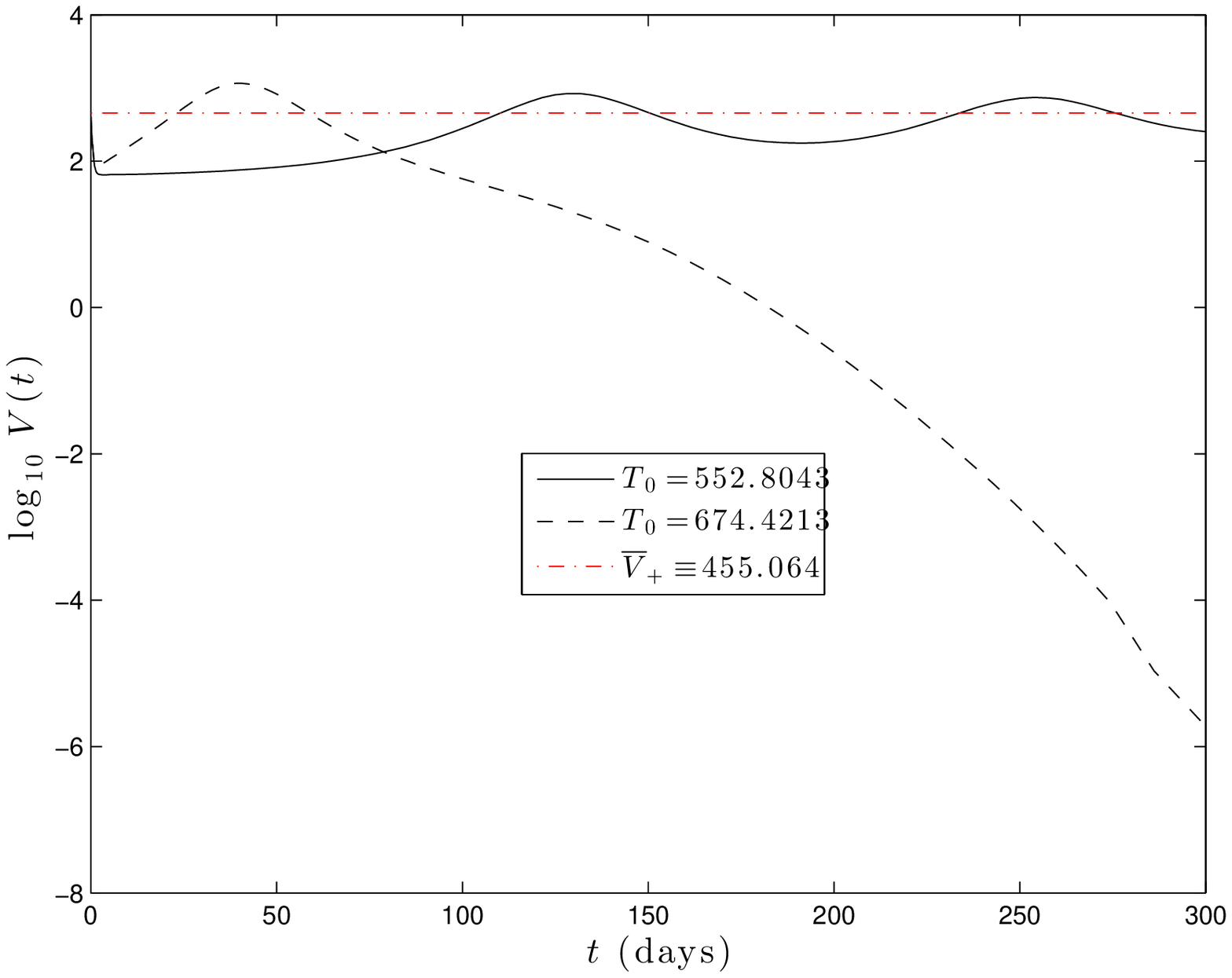}}\\
%\subfigure{\includegraphics[width=6.75cm]{../Figures/Stability/T0_change_V_4-29-16.eps}}\\
\vspace{-0.3in}
\hspace{-0.25in}
\subfigure{\includegraphics[width=6.75cm]{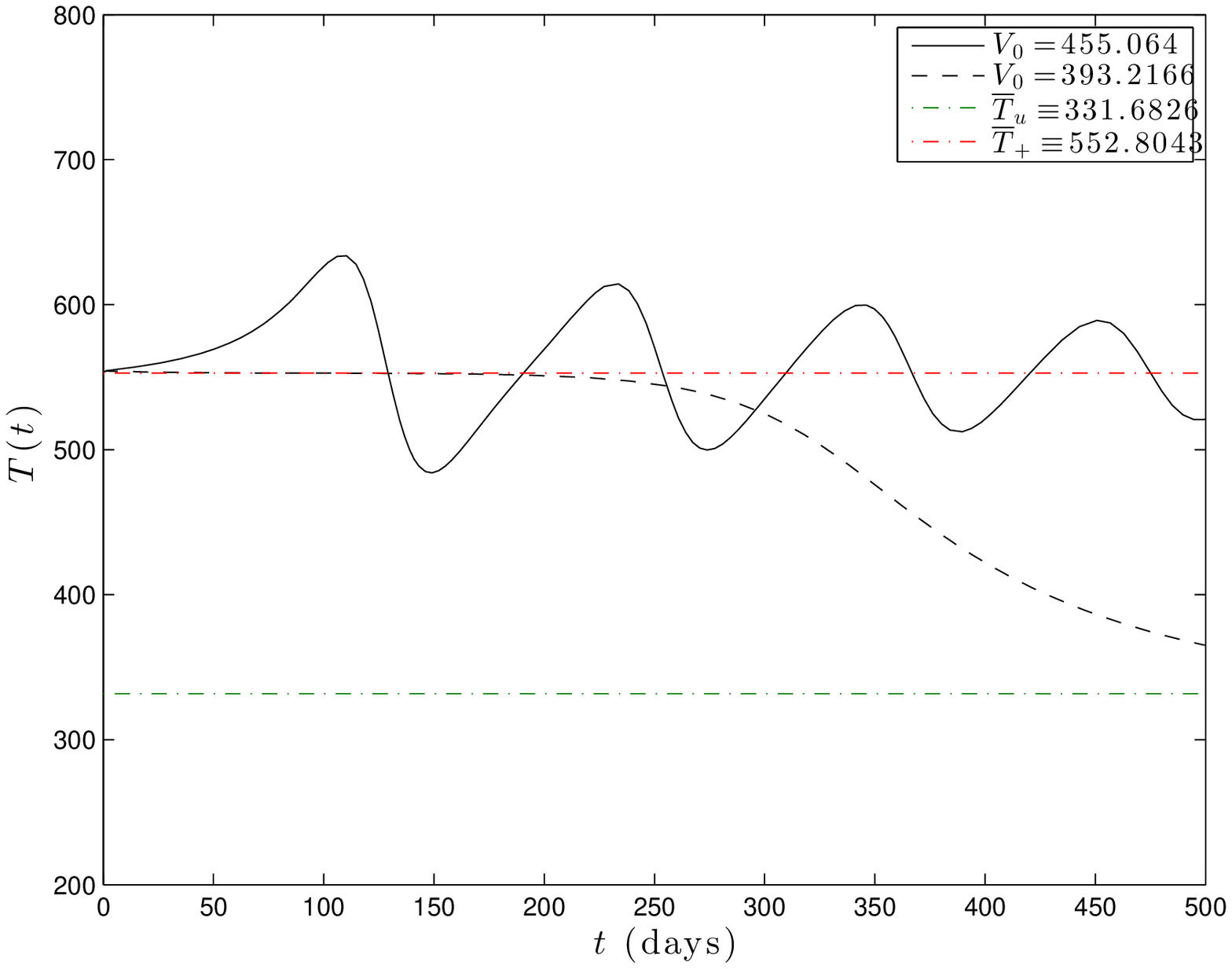}}
%\subfigure{\includegraphics[width=6.75cm]{../Figures/Stability/V0_change_T_4-29-16.eps}}
\hspace{-0.25in}
\subfigure{\includegraphics[width=6.75cm]{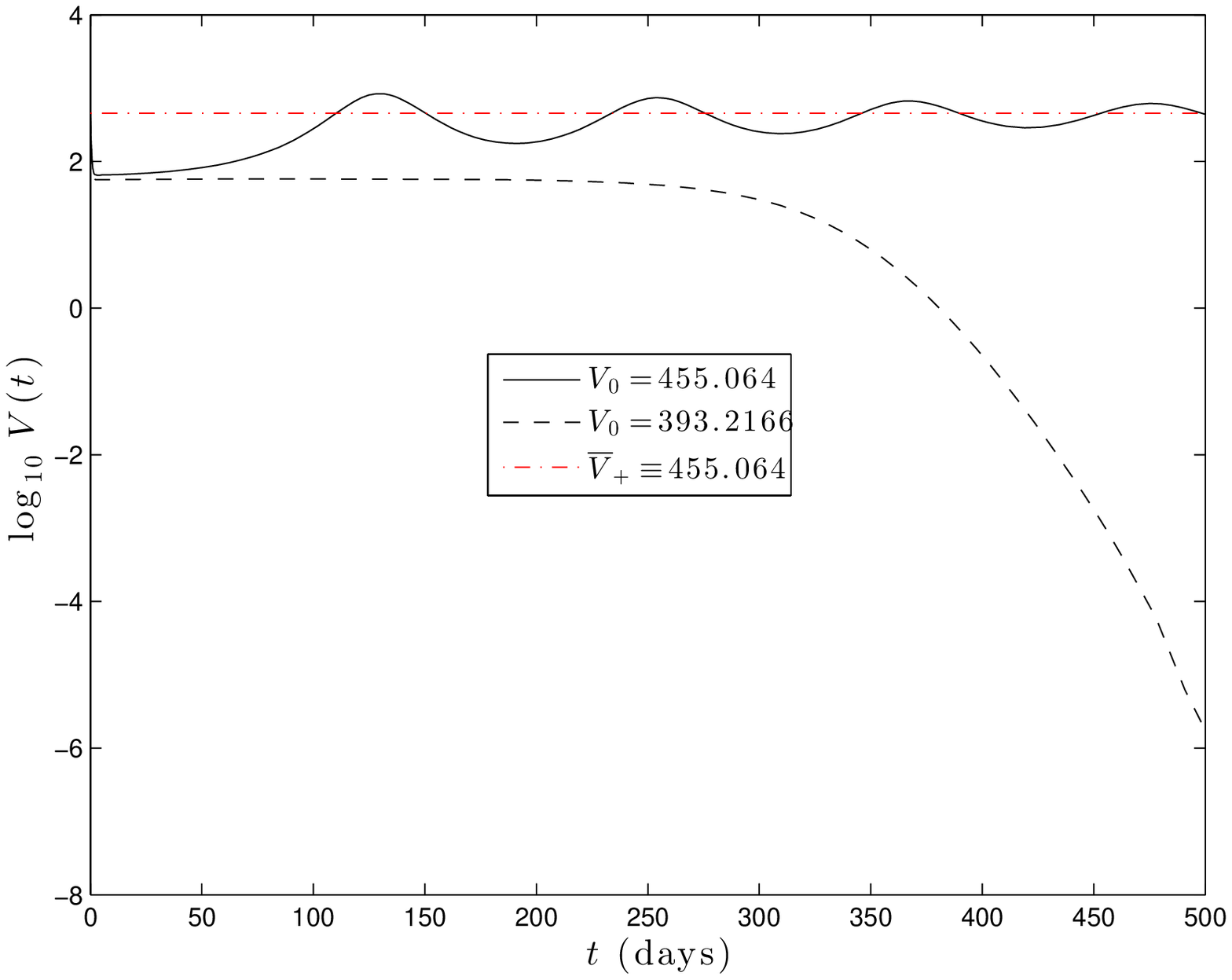}}
%\subfigure{\includegraphics[width=6.75cm]{../Figures/Stability/V0_change_V_4-29-16.eps}}
\vspace{-0.1in}
\caption{\footnotesize{Small changes in initial data, $T_0$ (top) and $V_0$ (bottom), yield changes in asymptotic behavior within the bistable parameter region  - T-cell count (left) and viral load (right). 
Within both simulations, dimensionless parameters are fixed to $R_m = 0.6$ and $R_0 = 3$. 
%Within the top simulations, the parameters are fixed to $R_m = 1.3$ and $R_0 = 0.87$***, while in the bottom simulations they are $R_m = 1.4$ and $R_0 = 0.7$***.
Recall that $\overline{V}_u \equiv 0$ is the equilibrium viral load for $E_u$, and note that $V(t)$ is represented on a log scale.}
}
\label{T0_V0_change}
\vspace{-0.1in}
\end{figure}

\begin{figure}[t] 
\centering
\includegraphics[height=.5\textwidth,width=.75\textwidth]{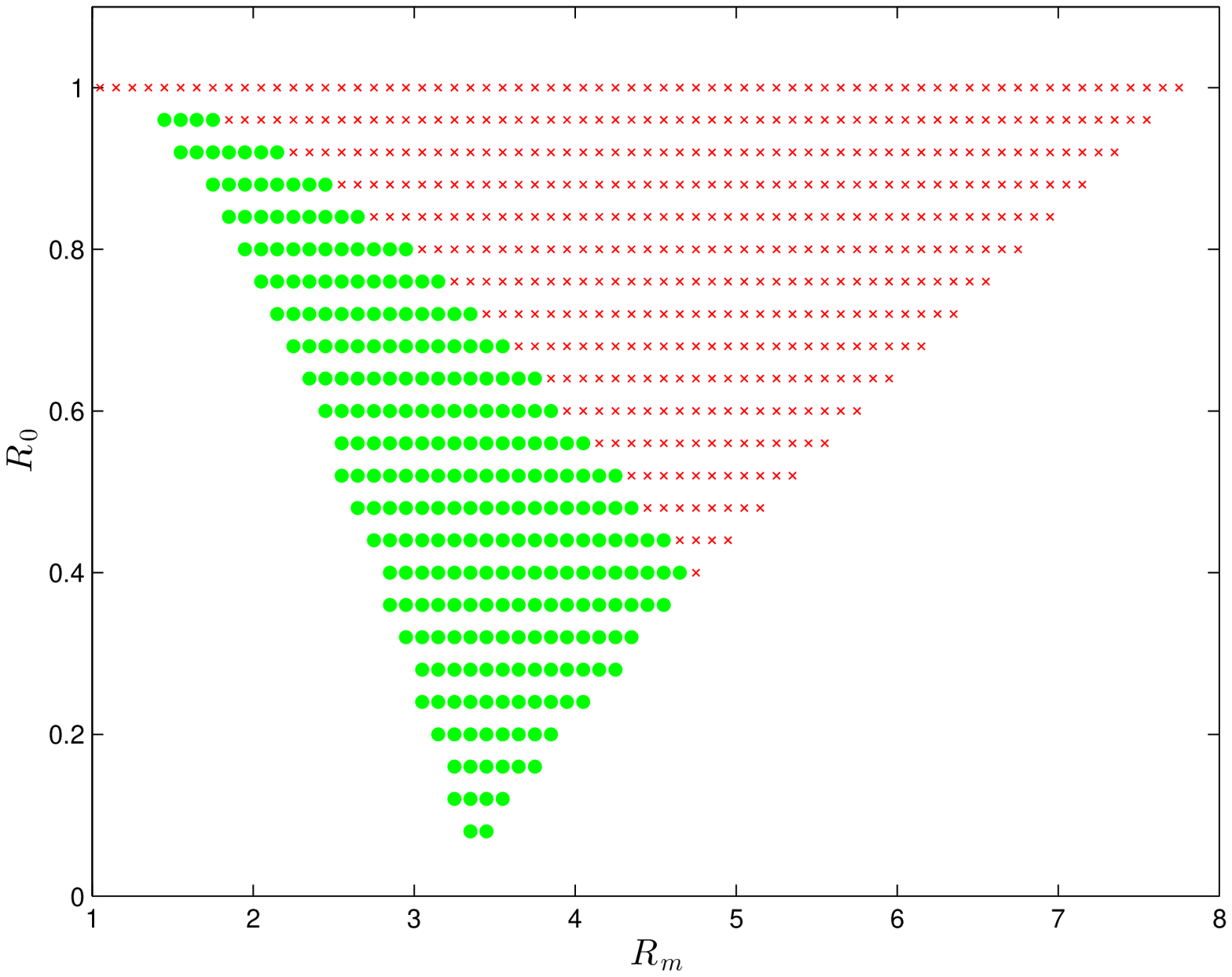}
\caption{\footnotesize{Simulation results for $R_m, R_0$ values within the bistable region.  A simulation was conducted at each point in the displayed parameter space with initial conditions $T_0 = 1, I_0 = 0, V_0 = \overline{V}_+$, where $\overline{V}_+$ is the equilibrium viral load of the $E_i^+$ infective equilibrium evaluated at $(R_m, R_0)$.  Green dots indicate solutions which tend toward the uninfected steady state $E_u$, while red crosses indicate solutions tending to $E_i^{+}$, the infected steady state.}}
\label{outcomes}
\end{figure}

%\begin{figure}[t]
%\centering
%\includegraphics[height=.5\textwidth]{../Figures/Changing_Initials/ChosenPoints_ColorCorrected.eps}
%\caption{The dependence of equilibria on initial values of healthy T-cells and the initial viral load at arbitrary points in the $(R_m,R_0)$ plane. Corresponding results are shown in Figure \ref{results}.}
%\label{locations}
%\end{figure}

\subsection{Basins of Attraction in the Bistable Region}
Since both steady states are locally stable in the overlapping region of Figure~\ref{las-stable}, we expect that differing long-term behavior, and hence different disease outcomes, may arise from variations in initial data.  Indeed, this is the case, and we demonstrate this by considering two different simulations of the model with $(R_m, R_0)$ values within the bistable region.
In order to display population values on a biologically pertinent scale, simulation values for $T(t)$ and $V(t)$ are referenced and displayed in the original, dimensional variables, rather than for the dimensionless system.

Figure~\ref{T0_V0_change} demonstrates the bistability of equilibria in two different scenarios. The first pair of simulations (Figure~\ref{T0_V0_change}, top row) fixes parameter values ($R_m = 0.6, R_0 = 3$) and the initial viral load (at $V_0 = \overline{V}_+$), but varies the initial T-cell count.  For $T_0 \approx 675$ the viral clearance ($E_u$) state is stable, and for $T_0 \approx 550$ the viral persistence ($E_i^+$) state is stable.  The top left plot shows T-cell count over time with the analytical steady state T-cell values, $\overline{T}_u \equiv 330$ and $\overline{T}_+ \equiv 550$, highlighted.  The top right plot shows the progression of infection, with the steady state viral load $\overline{V}_+ \approx 455$ highlighted (of course, $\overline{V}_u \equiv 0$).  We note that the large-time behavior of these solutions differs significantly even though the parameter values are identical and initial data are quite similar; in fact, the initial viral load is the same.

Similarly, the second pair of simulations (Figure~\ref{T0_V0_change}, bottom row) fixes parameter values %(for these simulations, $R_m = 1.4, R_0 = 0.7$) 
and the initial T-cell count (at $T_0 \approx 550$), with a varied initial viral load.  For $V_0 \approx 393$ the viral clearance ($E_u$) state is stable, and for $V_0 \approx 455$ the viral persistence ($E_i^+$) state is stable.  Again, steady-state T-cell values - $\overline{T}_u \equiv 331$ and $\overline{T}_+ \equiv 552$ - are highlighted within the bottom-left figure, while $\overline{V}_+ \equiv 455$, is shown in the bottom-right figure, and $\overline{V}_u \equiv 0$ for the uninfected equilibrium.  Again, the large-time behavior of these solutions differs significantly even though all initial and parameter values other than the initial viral load are equal.

To demonstrate this further, we simulate the progression of the model for varying $(R_m, R_0)$ values with other dimensionless parameters held constant and test whether the large-time behavior of these solutions tends towards the $E_u$ or $E_i^{+}$ steady state.  The results of these simulations over the bistable region is shown in Figure~\ref{outcomes}. 
%To explore the effects of changing initial conditions, we perform simulations on a variety of points within the bistable region (shown in Figure~\ref{locations}) over varying initial viral load and T-cell counts (shown in Figure~\ref{results}).

While these simulations provide information regarding the qualitative difference between solutions in this parameter region, they fail to describe how close initial data must be to equilibrium in order to guarantee their stability, i.e. the basin of attraction.  To study these basins of the two stable equilibrium states, we perform a number of perturbative computational studies at differing points in the bistable region.
The results, shown in Figures \ref{locations} and \ref{results}, display the sensitivity to initial conditions at each location.  Notice that locations closer to $R_0 = 1$ and further to the right in the bistable region display a greater basin of attraction for $E_i^+$ than those to the left of this region in the $(R_m,R_0)$ plane, as displayed by Figures \ref{results}(c), (e), and (f).  The appearance of a strip of persistent infection steady states increases while moving from left to right within this region.  We see that the initial viral load has a minor influence on the shape of the strip for various locations as the width of each strip is not uniform, while the initial T-cell count has a more pronounced affect on the long term behavior.  %The range of the initial healthy T-cell count which can alter the long term behavior is within $\pm 20\%$ or less of the equilibrium T-cell count of an infected person expressed within the dimensionless variables.
As these initial values are dimensionless, when $T_0$ is rescaled to represent an actual T-cell count, the basins of attraction for the infected steady state $E_i^+$ are increased by a factor of nearly $10^3$.  Similarly, the basin of attraction corresponding to perturbations in the viral load are increased by around $200$ when $V_0$ is rescaled to represent a true viral load.

An individual infected with HIV typically possesses baseline parameters (see Table \ref{List}) corresponding to a location in the $(R_m,R_0)$ plane lying above $R_0 = 1$, and hence within the $E_i^{+}$ stability region.  
However, certain parameter values, including the infection rate $k$ and the rate of viral production $p$, are known to vary widely amongst individuals, and thus the feasible region of attained values within the $(R_m,R_0)$ plane is quite vast.
In particular, $k$ has been reported to be as small as $10^{-6}$ \cite{Hern} and as large as $10^{-2}$ \cite{Hadji}, and this uncertainty could allow the values of $R_m$ and $R_0$ to vary within the range $10^{-2}-10^2$. 
Thus, we see that the biologically feasible parameter regime extends even into the unshaded region of Figure~\ref{las-stable}, which leads us to study the dynamics there, as well.

\begin{figure}[t]
\centering
\includegraphics[height=.5\textwidth]{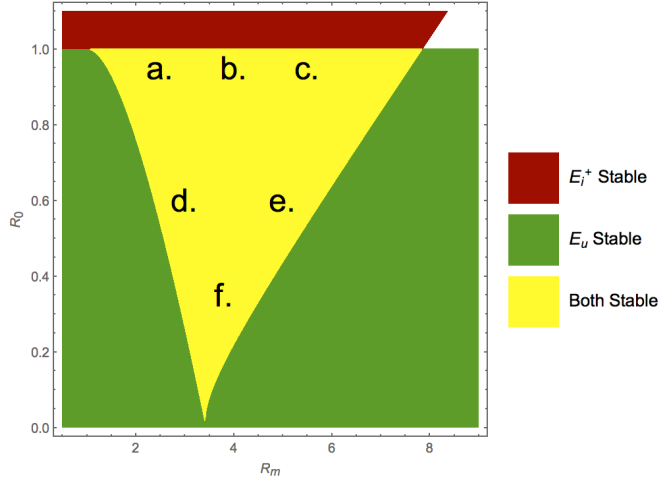}
\caption{\footnotesize{The dependence of equilibria on initial values of healthy T-cells and the initial viral load at arbitrary points in the $(R_m,R_0)$ plane. Corresponding results are shown in Figure \ref{results}.}
}
\label{locations}
\end{figure}

\begin{figure}[p] %old
\centering
\hspace{-0.2in}
%\subfigure[$R_m = 1.2, R_0 = 0.9$]{\includegraphics[width = 6.75cm]{../Figures/Changing_Initials/X_a_12_09.eps}}
\subfigure[$R_m = 2.5, R_0 = 0.9$]{\includegraphics[width = 6.75cm]{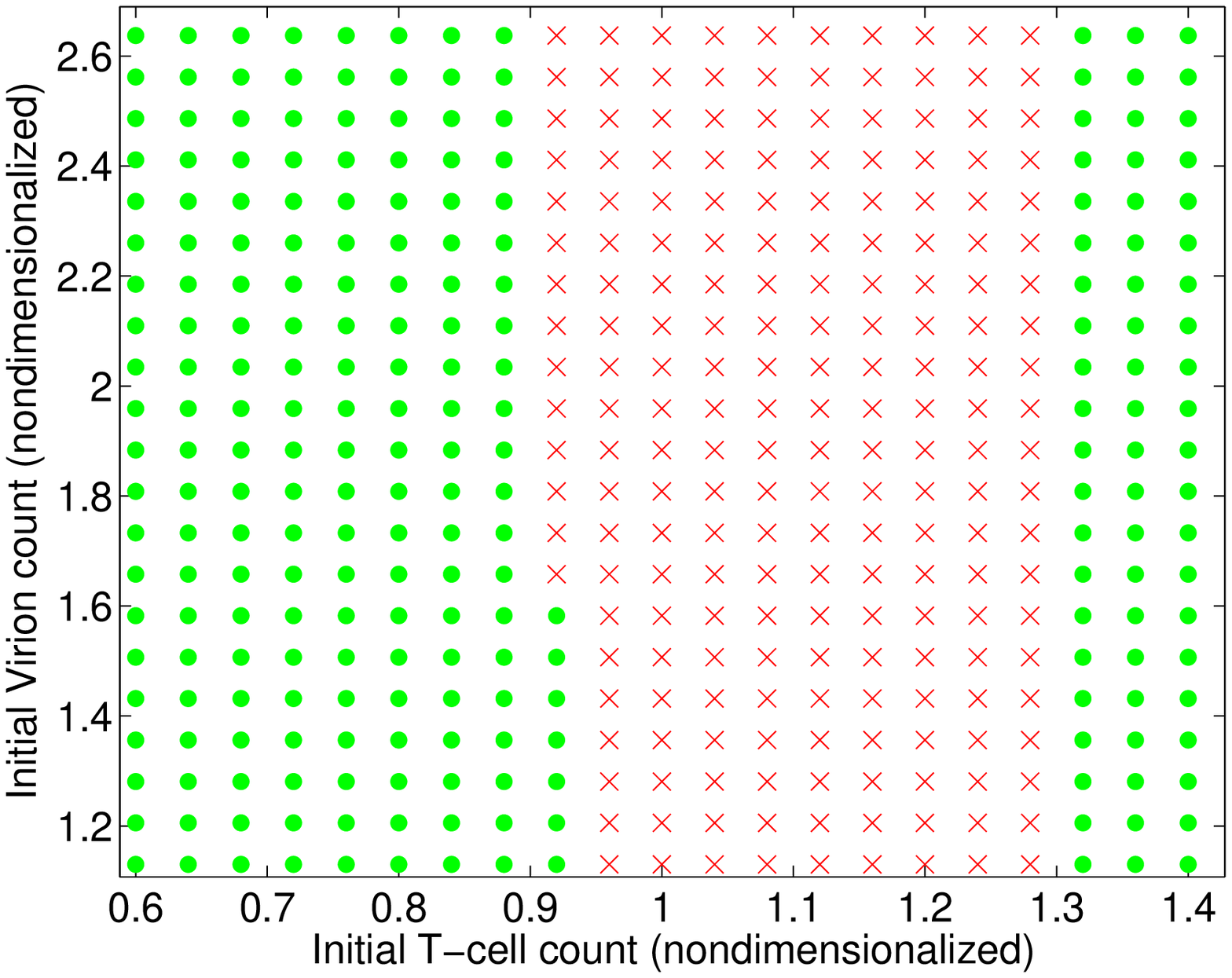}}
\hspace{-0.25in}
\subfigure[$R_m = 4, R_0 = 0.9$]{\includegraphics[width = 6.75cm]{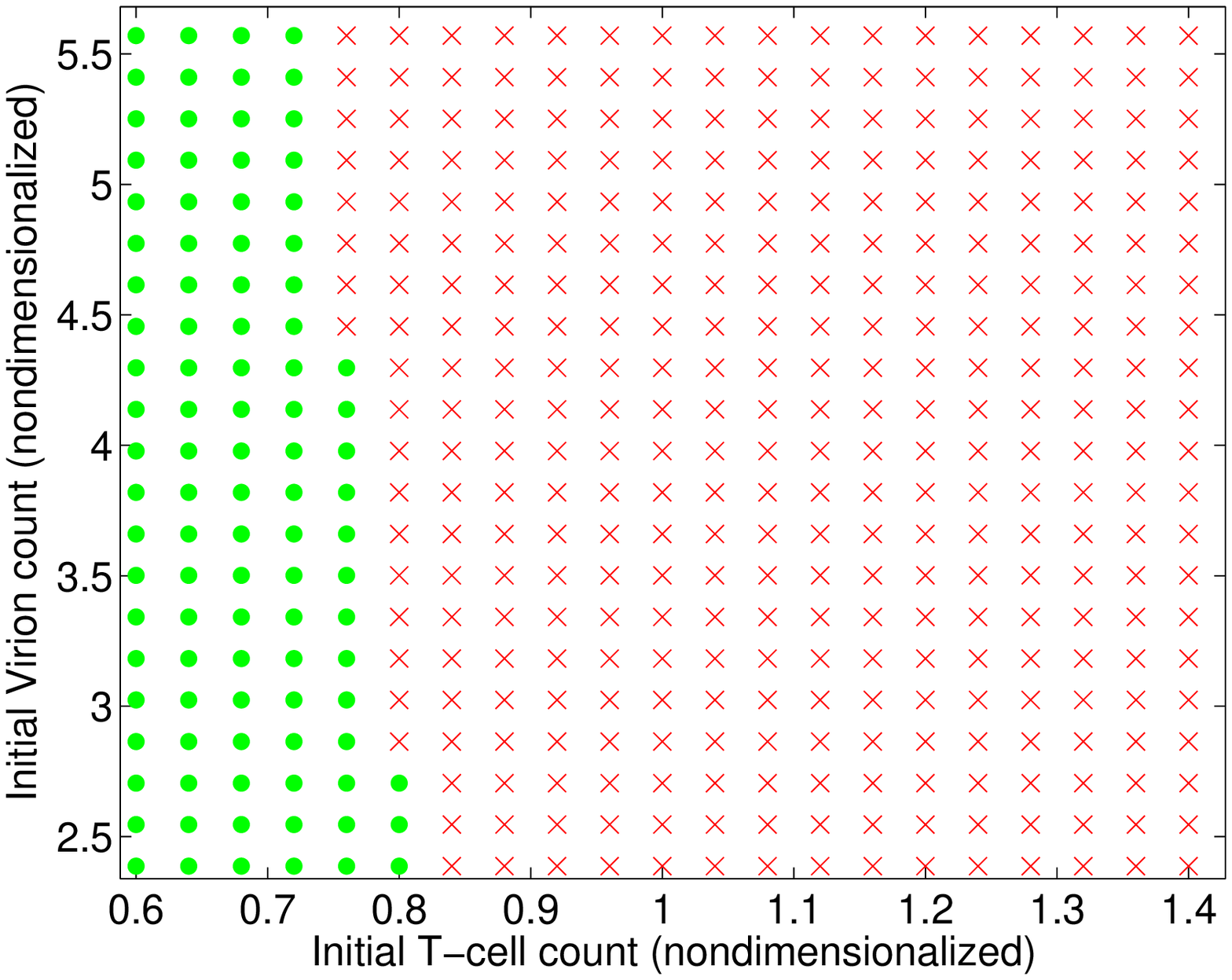}}\\
%\subfigure[$R_m = 1.35, R_0 = 0.9$]{\includegraphics[width = 6.75cm]{../Figures/Changing_Initials/X_b_135_09.eps}}\\
\vspace{-0.15in}
\hspace{-0.2in}
\subfigure[$R_m = 5.5, R_0 = 0.9$]{\includegraphics[width = 6.75cm]{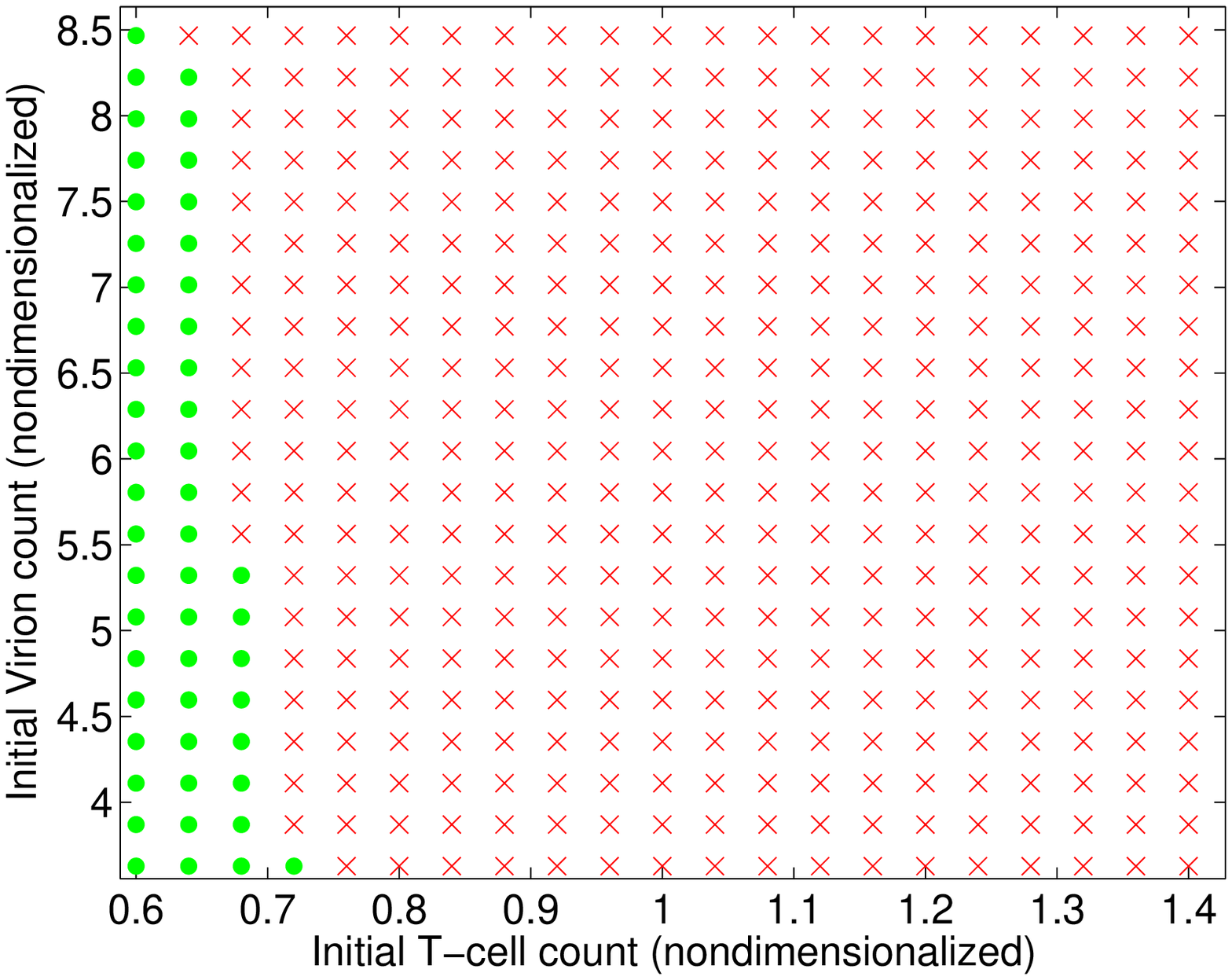}} 
%\subfigure[$R_m = 1.5, R_0 = 0.9$]{\includegraphics[width = 6.75cm]{../Figures/Changing_Initials/X_c_15_09.eps}} 
%Even expanded it can't show the full width of the band.
\hspace{-0.25in}
\subfigure[$R_m = 3, R_0 = 0.55$]{\includegraphics[width = 6.75cm]{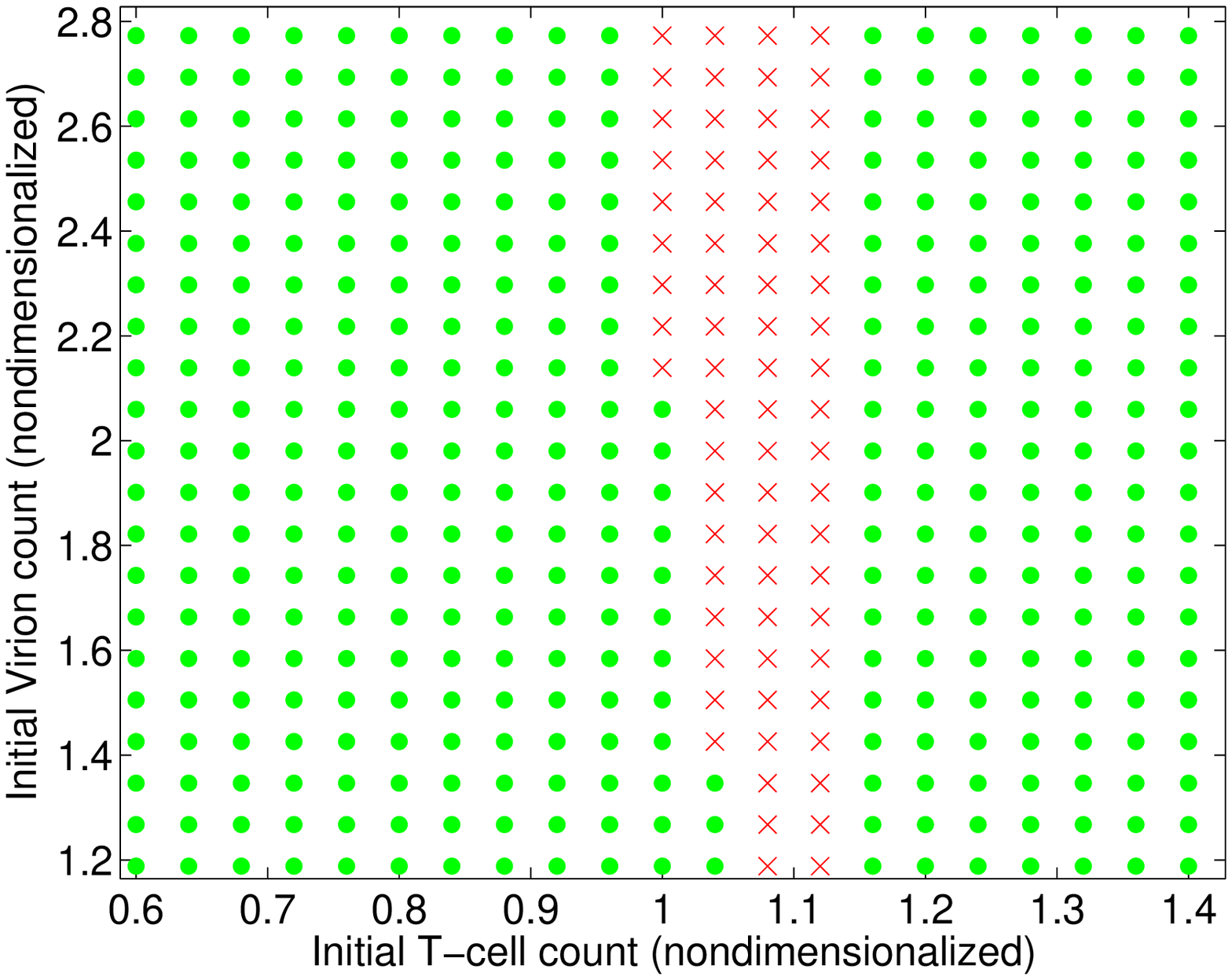}}\\ 
%\subfigure[$R_m = 1.35, R_0 = 0.55$]{\includegraphics[width = 6.75cm]{../Figures/Changing_Initials/X_d_135_055.eps}}\\ 
\vspace{-0.15in}
\hspace{-0.2in}
\subfigure[$R_m = 5, R_0 = 0.55$]{\includegraphics[width = 6.75cm]{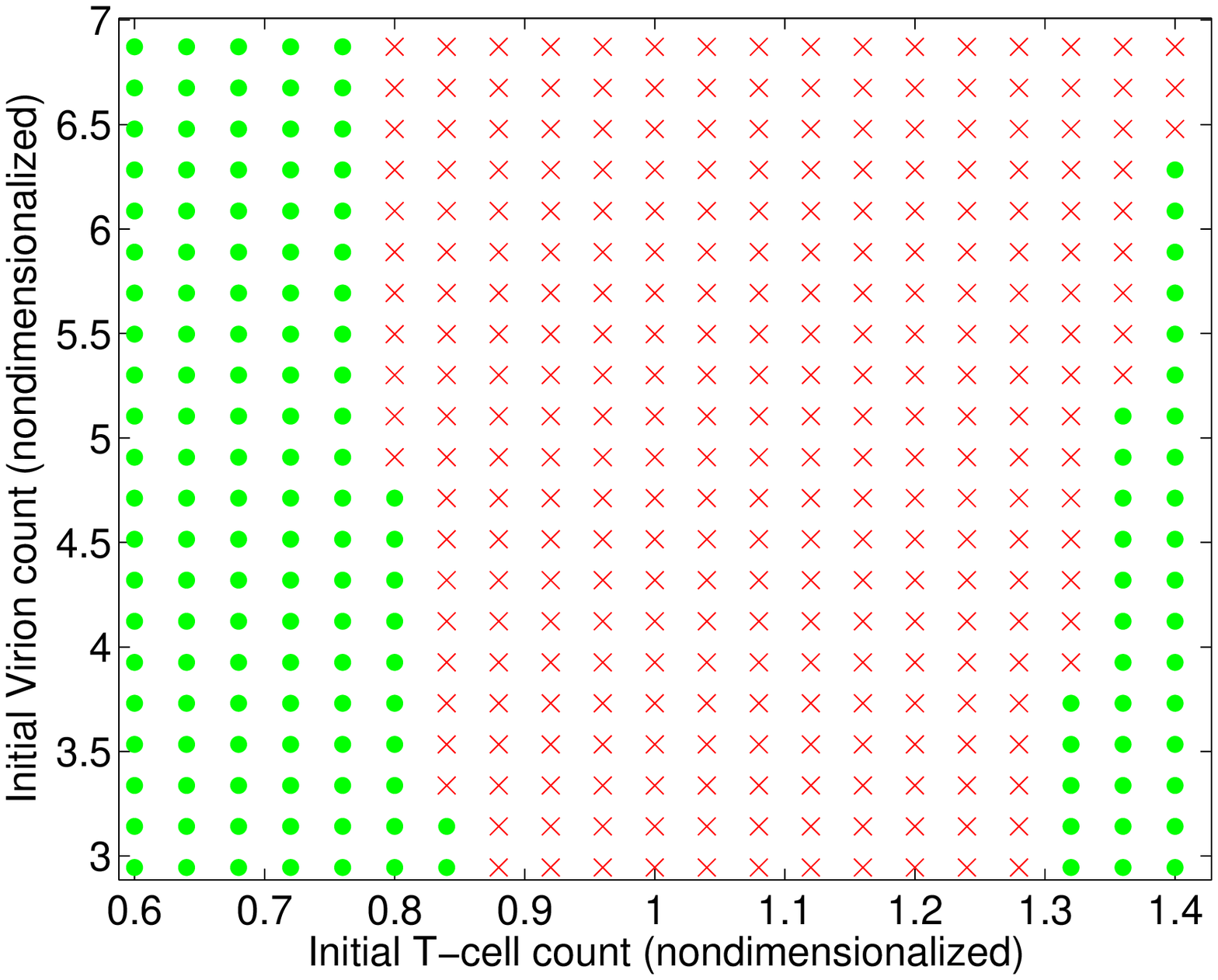}}
%\subfigure[$R_m = 1.5, R_0 = 0.55$]{\includegraphics[width = 6.75cm]{../Figures/Changing_Initials/X_e_15_055.eps}} %Requires an epsilon of 40% instead of 20% like the others (some even work in 10%, but 20% soon became the default - this is the only one I've found so far that is so bothered that it needs 40%)
\hspace{-0.25in}
\subfigure[$R_m = 4, R_0 = 0.2$]{\includegraphics[width = 6.75cm]{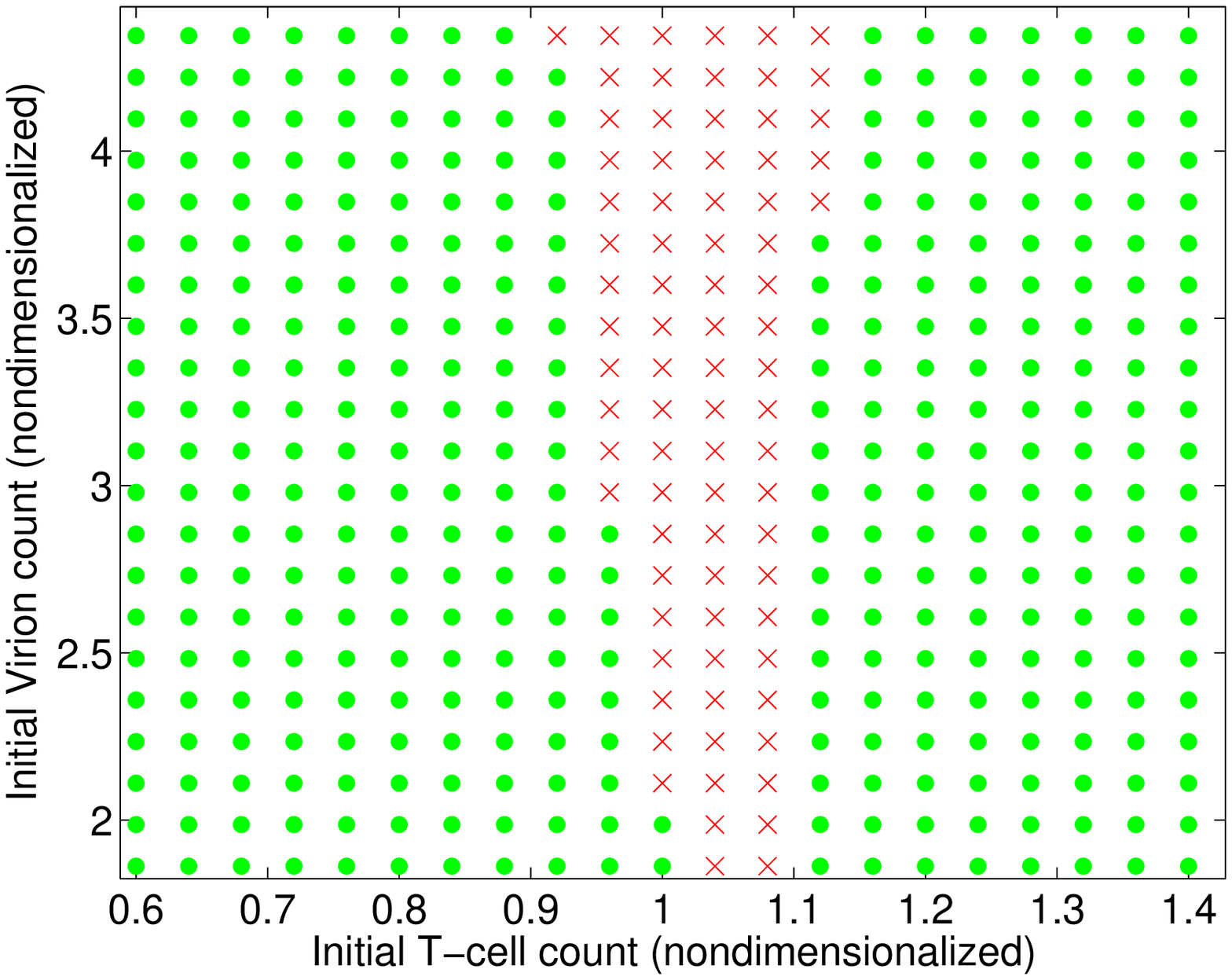}}
%\subfigure[$R_m = 1.45, R_0 = 0.2$]{\includegraphics[width = 6.75cm]{../Figures/Changing_Initials/X_f_145_02.eps}} %Also required higher epsilon (45% - and a longer run time)
\caption{\footnotesize{Basins of attraction within the bistable region. For each $(R_0, R_m)$ point, the initial T-cell count and viral load (both dimensionless) were set to the $E_i^+$ steady state evaluated at $(R_0, R_m)$. This forms the central point in each plot, and then both initial conditions are shifted $\pm 40\%$.  Green dots denote simulation convergence to $E_u$, and red crosses to $E_i^+$.
%The speed at which the solutions tend toward a steady state varied.  For points (e) and (f), longer simulation times and higher tolerances were needed to determine which state was attracted towards. 
%To provide greater detail to the small basin of attraction, simulations at point (f) use a nonzero initial infected T-cell count and vary only $\pm 10\%$.
%To provide greater detail, simulations at point (f) use a smaller shift ($\pm 10 \%$) and a nonzero initial infected T-cell count. %(not near the magnitude of the equilibrium value).
A visual summary of the points chosen in the bistable region is provided by Figure~\ref{locations}.}
}
\label{results}
\end{figure}

\subsection{Hopf Bifurcation}
From Figure \ref{las-stable}, we notice that there is a region of the parameter space within which no equilibrium point is locally stable.  Thus, one may expect that a different attracting set inherits this property for such parameter values.  Indeed, this is the case, and as we will show using $R_m$ as a bifurcation parameter, a Hopf bifurcation occurs at the boundary of this domain.
In particular, we will take any $R_0 > 1$, vary $R_m$ to move within this region of the $(R_m,R_0)$ plane, and investigate the stability properties of $E_i^+$ as they change along the right boundary of the (dark) red region within Figure \ref{las-stable}. Since $R_m$ will be used to move through the parameter space, we will alter notation when necessary in order to denote certain quantities that depend on this parameter.
As demonstrated within Appendix~\ref{appC}, the Jacobian of system \eqref{AcuteNd} evaluated at $E_i^+$ is
\begin{equation*}
\begin{large}
\centering
	\nabla f(E^+_i) = \left(
	\def\arraystretch{2}\begin{array}{ccc}
		-R_0 & 0 & \frac{R_m}{(1+ \beta \overline{V}_+)^2} - 1 \\
		\alpha_1 \overline{V}_+ & -\alpha_1 &  \alpha_1\\
		0 & \alpha_2 & -\alpha_2 
\end{array} \right)
\end{large}
\end{equation*}
where
$\overline{V}_+ = \overline{V}_+(R_m)$ is the $E_i^+$ viral population given by Theorem 3.1.  Recall that this steady state value satisfies the quadratic equation
\begin{equation}
\label{Vbar}
\beta \overline{V}_+^2 + \left [ \beta(1-R_0) + 1 - R_m \right ]\overline{V}_+ + 1 - R_0= 0
\end{equation}
and thus varies with $R_m$ when all other parameters are fixed.
The characteristic polynomial associated to $\nabla f(E^+_i)$ is 
\begin{equation}
\label{charpoly}
\eta^{3} + d_2 \eta^{2} + d_1 \eta + d_0(R_m) = 0
\end{equation} 
where
$$\begin{gathered}
d_0(R_m) = \frac{\alpha_1\alpha_2 \overline{V}_+}{(1 + \beta \overline{V}_+)^2} \left [(1 + \beta \overline{V}_+)^2 - R_m \right ],\\
d_1 = (\alpha_1 + \alpha_2)R_0,\\
d_2 = \alpha_1 + \alpha_2 + R_0.
\end{gathered}$$
Using \eqref{Vbar}, and in particular the relationship
$$ (1-R_m) \overline{V}_+ = -\beta \overline{V}_+^2 + (R_0 -1) (1 + \beta \overline{V}_+),$$
we can rewrite $d_0$ as 
\begin{equation}
\label{d0}
d_0(R_m) = \frac{\alpha_1\alpha_2}{1 + \beta \overline{V}_+} \left [\beta \overline{V}_+^2 + R_0 - 1 \right ].
\end{equation}
%Note that $d_0$ is denoted by $d_0(R_0)$ throughout this section in order to emphasize its dependence on $R_0$.
%
%d_0(R_0) = \frac{\alpha_1\alpha_2}{(1 + \beta V)^2} \left [(1 + \beta V)^2 - R_m \right ],\\
%
Notice that $d_1, d_2 > 0$ since all parameters are positive.
Additionally, we establish the following result, which will be useful in our study of the roots of \eqref{charpoly} and in the proof of Theorem \ref{T1}.
\begin{lemma}
Within the region of existence of $E_i^{+}$, we have 
$$ \beta \overline{V}_+^2 + R_0 - 1 > 0$$
and thus $d_0(R_m) > 0$ within this parameter region.
\label{d0pos}
\end{lemma}

\begin{figure}
\hspace{-0.2in}
\subfigure[Real part of eigenvalues $\eta_2$ and $\eta_3$]{\includegraphics[width=6.75cm]{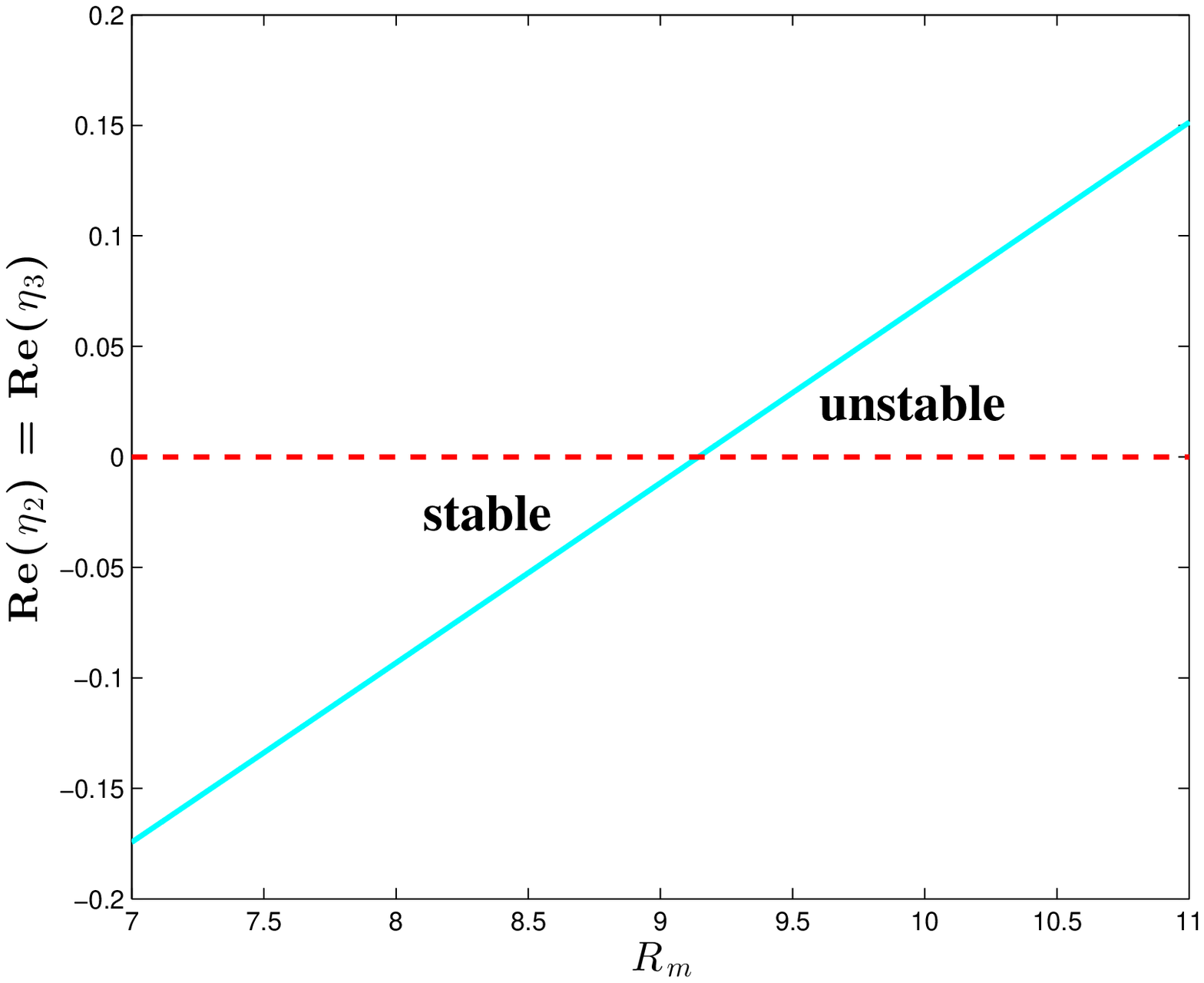}}
\hspace{-0.25in}
\subfigure[Imaginary part of eigenvalues $\eta_2$ and $\eta_3$]{\includegraphics[width=6.75cm]{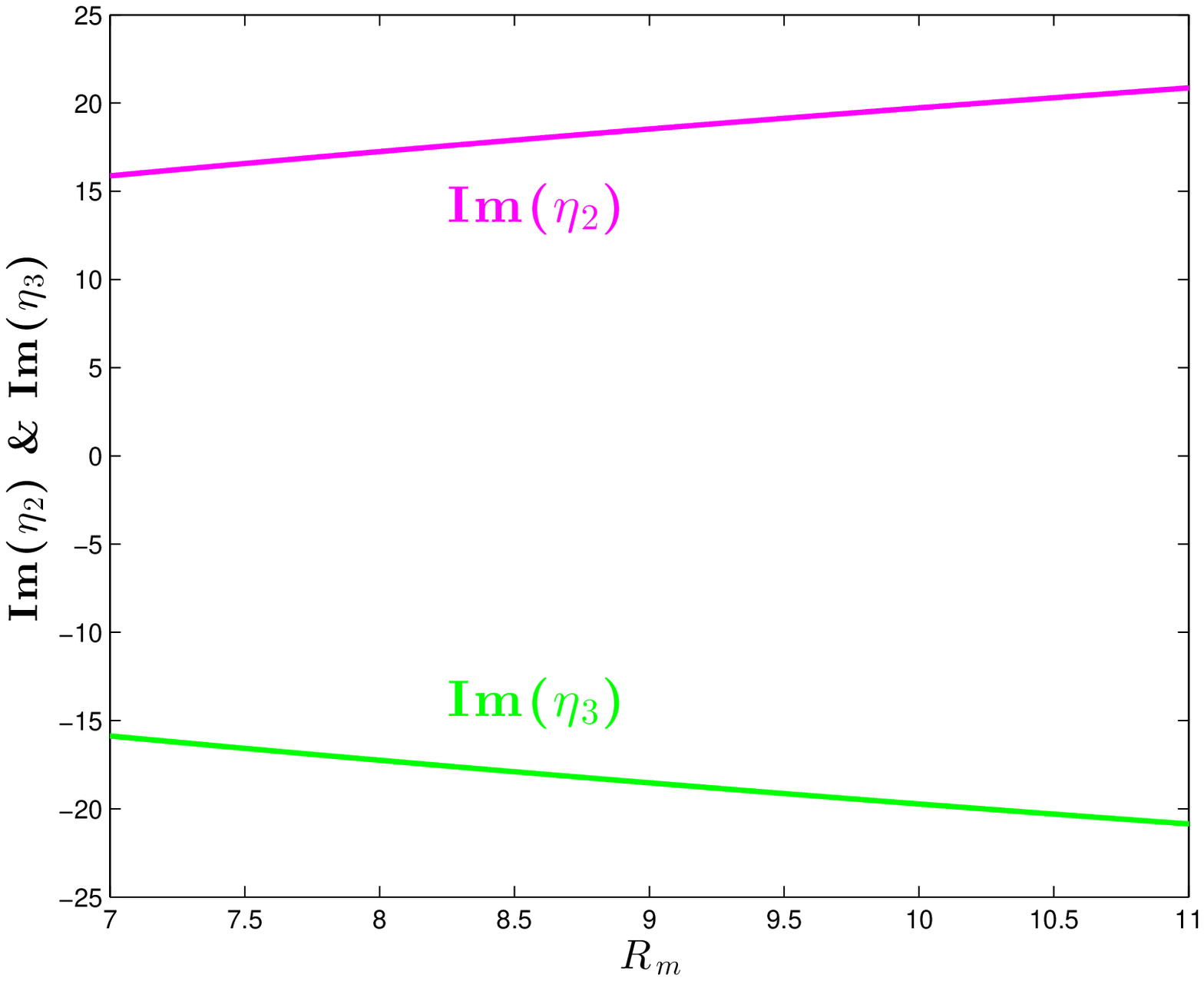}}
\caption{\footnotesize{Real and imaginary parts of eigenvalues $\eta_2$ and $\eta_3$ as functions of the bifurcation parameter $R_m$ with $R_0 = 1.25$ fixed.  These curves cross the imaginary axis at $R_m = R_m^{*} \approx 9$ when using parameter values from Table \ref{List}.}
}
\label{eigens}
\end{figure}

The proof of Lemma \ref{d0pos} is contained within Appendix \ref{appB}.
%, as shown in Figure \ref{exist}, and specifically within the yellow region in Figure \ref{stable}(a).  
Next, denote the corresponding roots of the characteristic polynomial \eqref{charpoly} by $\eta_i(R_m)$, $i = 1, 2, 3$.  It follows that the mapping $R_m \rightarrow \eta_i(R_m)$ is smooth, as displayed within Figure \ref{eigens}.
To begin our study of the eigenvalues of this system, we first show the existence of a negative real root within the local stability region of $E_i^+$.
\begin{lemma}
 For any $R_m > 0$ within the region of existence for $E_i^{+}$, the characteristic polynomial possesses at least one real, negative root. Additionally, all real roots are negative.
\label{negreal}
\end{lemma}
\begin{proof}
By the Fundamental Theorem of Algebra, the above characteristic polynomial will have exactly three roots.  Thus, it will either possess one real root and two complex roots, or three real roots.
Using Lemma \ref{d0pos}, we find $d_0, d_1, d_2 > 0$ in the region of existence for $E_i^{+}$, and thus any real root must be negative. Hence, the characteristic polynomial has either three negative real roots or one negative real root and two complex conjugate roots. In either case, the conclusions follow.  
\end{proof}

Next, we define
$$D_2(R_m) = d_1d_2 - d_0(R_m)$$
to be the second Hurwitz determinant of the characteristic polynomial \eqref{charpoly}.
Further, for a fixed value of $R_0 > 1$, let $R_m^*$ be the value of $R_m$ in the region of existence of $E_i^{+}$ such that $D_2(R_m^*) = 0$.
We will see that the curve $D_2(R_m) = 0$ corresponds to the right edge of the $E_i^+$ stability region in Figure \ref{las-stable}.
The following result shows that as the value of $R_m$ is increased beyond $R_m^*$, so that the point $(R_m, R_0)$ lies outside of the $E_i^+$ stability region, the complex eigenvalues become purely imaginary.
\begin{lemma}
At the point $R_m = R_m^*$, two eigenvalues of $\nabla f (E_i^{+})$, denoted $\eta_2(R_m^*)$ and $\eta_3(R_m^*)$, are purely imaginary and conjugate, while the third, $\eta_1(R_m^*)$, is real and negative.
\label{R0star}
\end{lemma}

\begin{proof}
At the value $R_m = R_m^*$,  we have $D_2(R_m^*) = 0$ and thus $d_0(R_m^*) = d_1 d_2$.  Therefore, \eqref{charpoly} can be written as
$$\eta^{3} + d_2 \eta^{2} + d_1 \eta + d_1d_2 = 0$$
or
$$(\eta^2 + d_1) ( \eta + d_2)  = 0,$$
\noindent which possess the roots
\begin{align*}
\eta_1 &= -d_2 = - (\alpha_1 + \alpha_2 + R_0)\\
\eta_2 &= i \sqrt{d_1} = i \sqrt{(\alpha_1 + \alpha_2) R_0}\\
\eta_3 &= - i \sqrt{d_1} = - i \sqrt{(\alpha_1 + \alpha_2) R_0}.
\end{align*}
\end{proof}

%\begin{figure}[t] %old?
%\centering
%\includegraphics[height=.5\textwidth]{bprime.eps}
%\caption{The region in the $(R_m, R_0)$-plane where $B'(\rho) < 0$ includes the locations where $\rho = \rho^{*}$.}
%\label{bprime}
%\end{figure}

Finally, with this understanding of $\eta_1$, $\eta_2$, and $\eta_3$, we can prove the existence of a Hopf bifurcation across the curve $D_2(R_m) = 0$.

\begin{theorem}
For $R_0 > 1$, a Hopf bifurcation occurs at the critical value $R_m = R_m^*$.  In particular, as the value of $R_m$ crosses $R_m^*$, the equilibrium point $E_i^{+}$ becomes unstable and a stable limit cycle branches from the equilibrium.
\end{theorem}

\begin{proof}
With Lemmas \ref{negreal} and \ref{R0star}, we need only show the transversality condition, $\frac{d \eta_2}{dR_m}(R_m^*) \neq 0$ to prove the existence of a Hopf bifurcation.  Using the characteristic polynomial \eqref{charpoly} evaluated at $\eta_2(R_m)$, we take the $R_m$-derivative to find
$$\frac{d \eta_2}{dR_m} \left[ 3\eta_2^2 + 2d_2 \eta_2 + d_1\right] + d_0^\prime(R_m) = 0$$
which yields
$$\frac{d \eta_2}{dR_m}(R_m) = -\frac{d_0'(R_m)}{3 \eta_2^{2}+2 d_2\eta_2 + d_1}.$$
Evaluating the derivative at $R_m = R_m^*$ and substituting the known value for 
$\eta_2(R_m^*) = i \sqrt{d_1}$
this becomes
$$\frac{d \eta_2}{dR_m}(R_m^*) = \frac{ -d_0'(R_m^*)}{-2 d_1 +2\textit{i} d_2\sqrt{d_1}}.$$
Thus, multiplying by the conjugate and taking the real part, we find
$$\mathfrak{Re}\left(\frac{d \eta_2}{dR_m} (R_m^*)\right) = \frac{d_0'(R_m^*)}{2[d_1 + d_2^2]}.$$
From this, it follows that $\mathfrak{Re}\left(\frac{d \eta_2}{dR_m}(R_m^*)\right) \neq 0$ if and only if $d_0'(R_m^*) \neq 0$.
Using \eqref{d0}, we compute this term as
$$d_0'(R_m) = \frac{\beta}{1 + \beta \overline{V}_+} \left [ -d_0(R_m) + 2\alpha_1\alpha_2 \overline{V}_+ \right] \frac{d\overline{V}_+}{dR_m}.$$ 
Now, using \eqref{Vbar} we compute $\frac{d\overline{V}_+}{dR_m}$ so that
$$2\beta \overline{V} \frac{d\overline{V}_+}{dR_m} - \overline{V}_+ + \left [ \beta(1 - R_0) + 1-R_m \right ] \frac{d\overline{V}_+}{dR_m} = 0,$$
and after some algebra and use of \eqref{Vbar}, this implies
$$\frac{d\overline{V}_+}{dR_m} = \frac{\overline{V}_+}{ 2\beta \overline{V}_+ + \beta(1 - R_0) + 1-R_m} = \frac{\overline{V}_+^2}{\beta \overline{V}_+^2 + R_ 0 - 1}.$$
Since $\overline{V}_+ > 0$ for $R_0 > 1$, we conclude that this term is strictly positive.
Finally, a brief computation using \eqref{d0} shows that the remaining term in $d_0'(R_m)$ satisfies
$$ - d_0(R_m) + 2\alpha_1\alpha_2 \overline{V}_+ = \frac{\alpha_1\alpha_2 \overline{V}_+}{1 + \beta \overline{V}_+} \left [R_m + 1 + \beta (R_0 - 1) \right ].$$
Since $R_0 > 1$ and $\overline{V}_+ > 0$, this term is strictly positive, as well.  Hence, we find $d_0'(R_m^*) \neq 0$ and the proof is complete.
%In particular, since these terms are positive, the real part of the eigenvalue is increasing across $R_m = R_m^*$.  By Theorem \ref{R0star}, the real part is exactly zero at $R_m = R_m^*$, which implies their positivity as $R_m$ is further increase, and hence the instability of the periodic orbit.
\end{proof}

To supplement these analytical results, we also include representative simulations of the system for specific values of $R_0>1$ near the point $R_m = R_m^*$ within Figures \ref{stableEi} - \ref{largeorbit}.  In particular, Figure \ref{stableEi} demonstrates the stability of the infected steady state within the region in which only ${E_i}^{+}$ is locally asymptotically stable.  In such a case ($R_0 = 1.25$, $R_m = 6.5$ here), the complex eigenvalues of the Jacobian matrix evaluated at  ${E_i}^{+}$ possess negative real part, and Figure \ref{stableEi}(a) displays the dynamics of the corresponding solution spiraling inward toward the $E_i^+$ equilibrium. In Figure \ref{smallorbit}, the bifurcation parameter $R_m$ is adjusted so that $R_m \approx R_m^{*}$ and the location in the parameter plane is at the border of the region of $E_i^+$ stability and the unshaded region seen within Figure \ref{las-stable}. The associated complex eigenvalues now have real parts which approach zero, and the emergence of a periodic orbit in Figure \ref{smallorbit}(a) becomes more visible. Finally, $R_m$ is further increased so that the parameter plane location lies within the unshaded region.  The complex eigenvalues now have positive real part, and the solution settles into a periodic orbit as seen within Figure \ref{largeorbit}.

\begin{figure} %old
\hspace{-0.2in}
\subfigure[TIV phase portrait]{\includegraphics[height = 5cm,width =6.75cm]{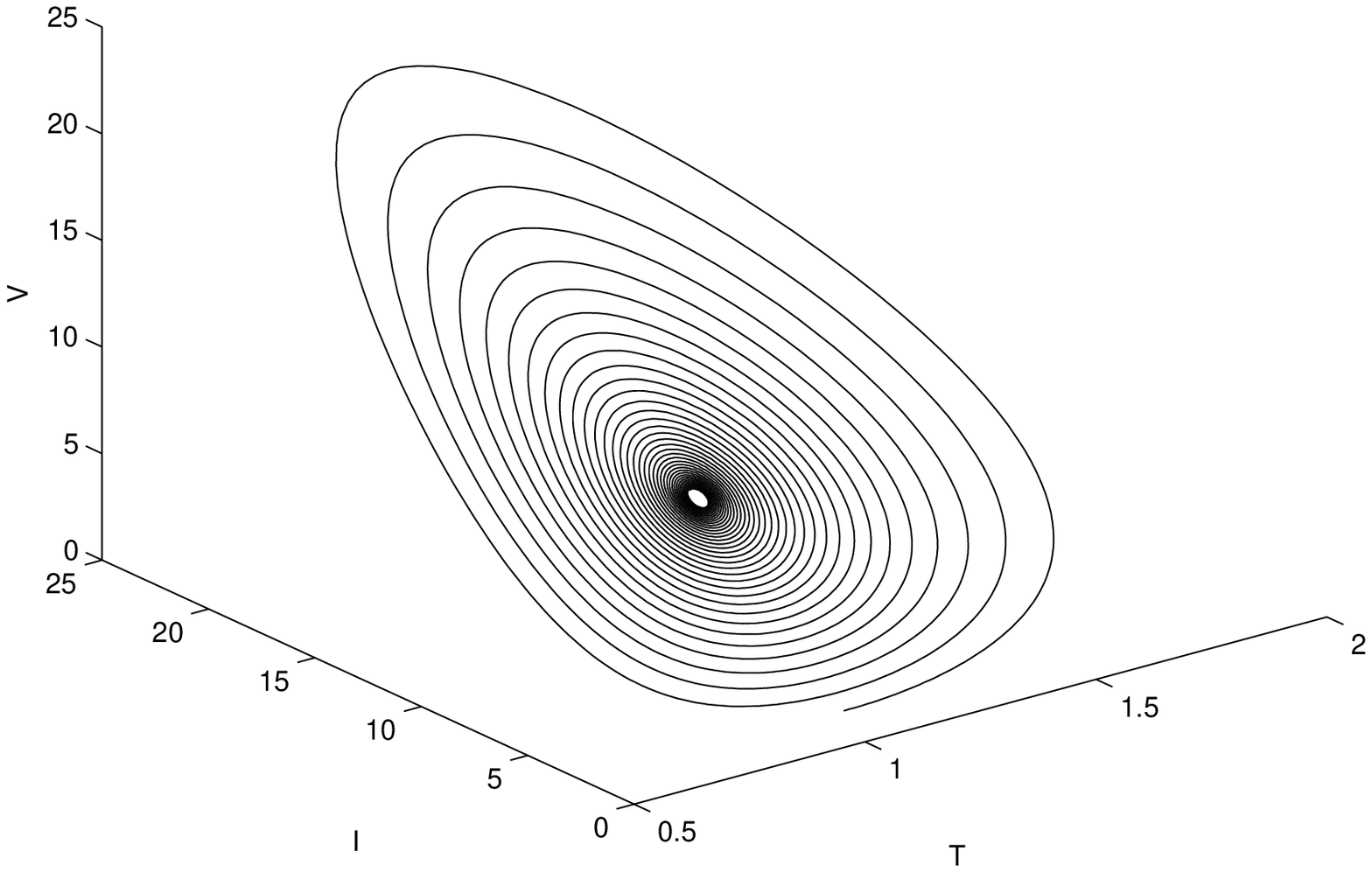}}
\hspace{-0.25in}
\subfigure[Uninfected T-cell population]{\includegraphics[height = 5cm,width =6.75cm]{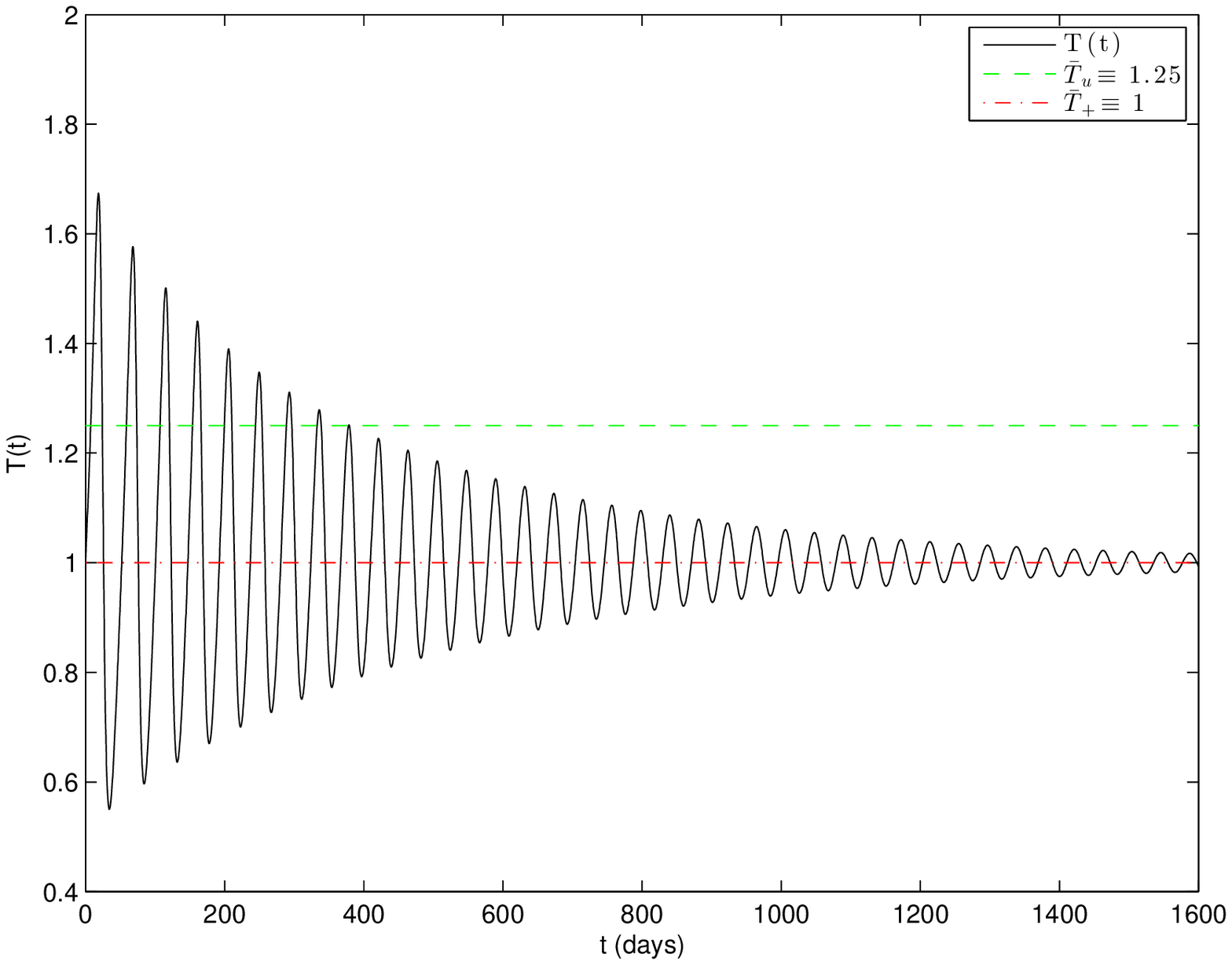}}\\
\hspace{-0.2in}
\subfigure[Virus population]{\includegraphics[height = 5cm,width =6.75cm]{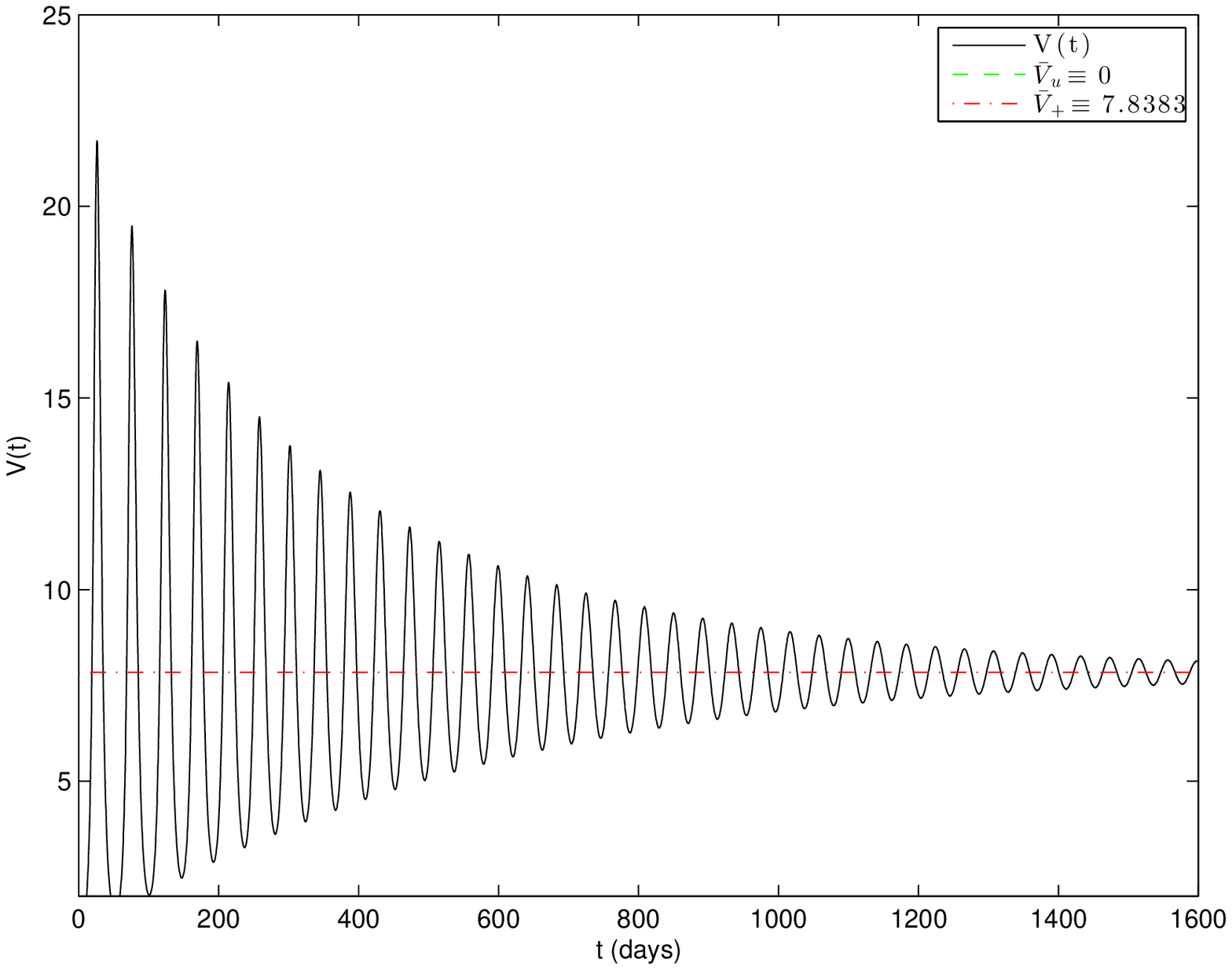}}
\hspace{-0.25in}
\subfigure[Infected T-cell population]{\includegraphics[height = 5cm,width =6.75cm]{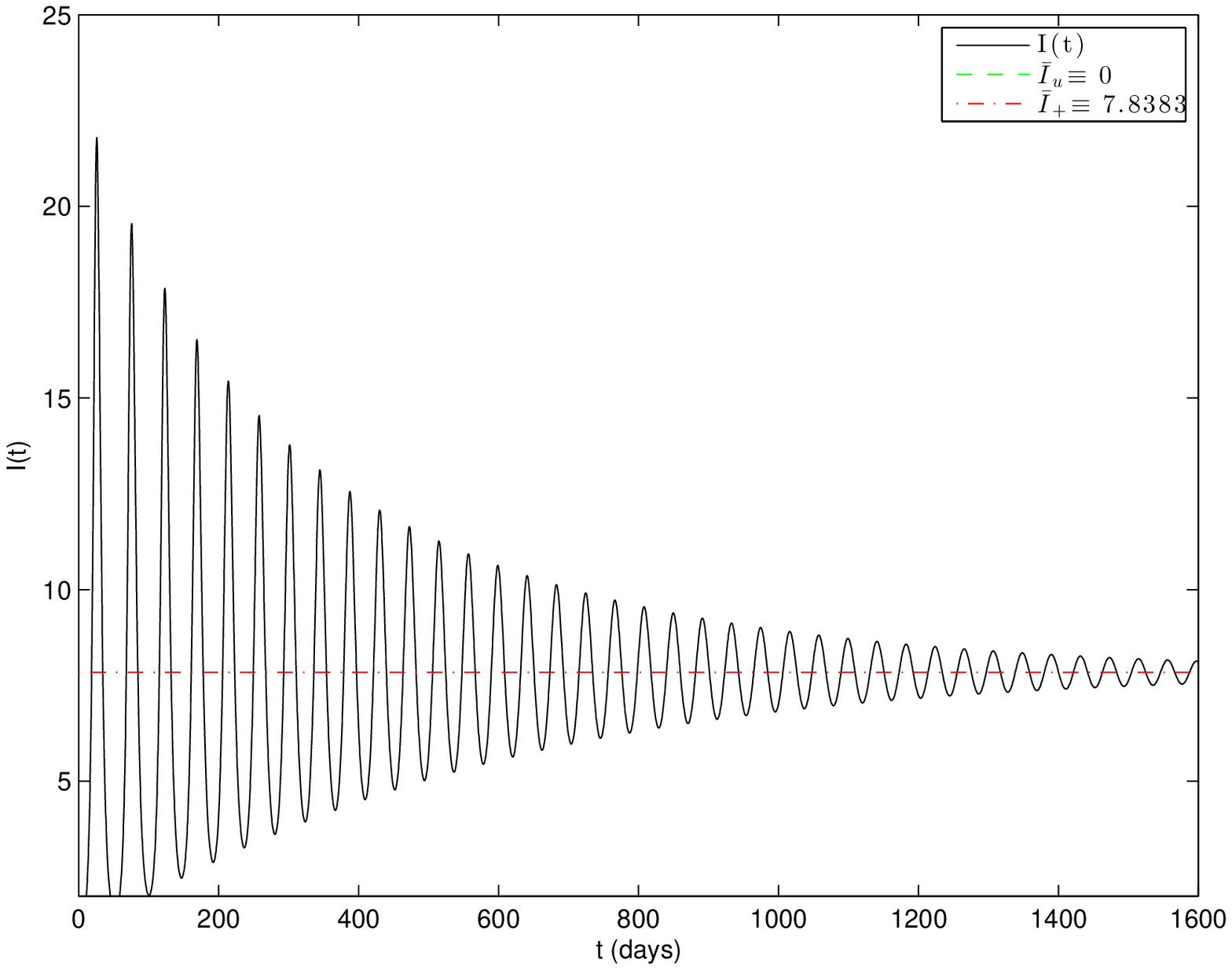}}
\caption{\footnotesize{Simulation of the dimensionless system for $R_m = 6.5$ and $R_0 = 1.25$. The parameter location is in the (dark) red region of Figure \ref{las-stable}. Hence, complex eigenvalues of the system linearized about $E_i^+$ are of the form $\alpha \pm i \beta$ where $\alpha < 0$, yielding a stable steady state. Green and red horizontal lines indicate corresponding population values of the $E_i^+$ and $E_u$ steady states, respectively.}
}
\label{stableEi}
\end{figure}

\begin{figure} %old
%% COMMENTED FOR SIZE:
\hspace{-0.2in}
\subfigure[TIV phase portrait]{\includegraphics[height = 5cm,width =6.75cm]{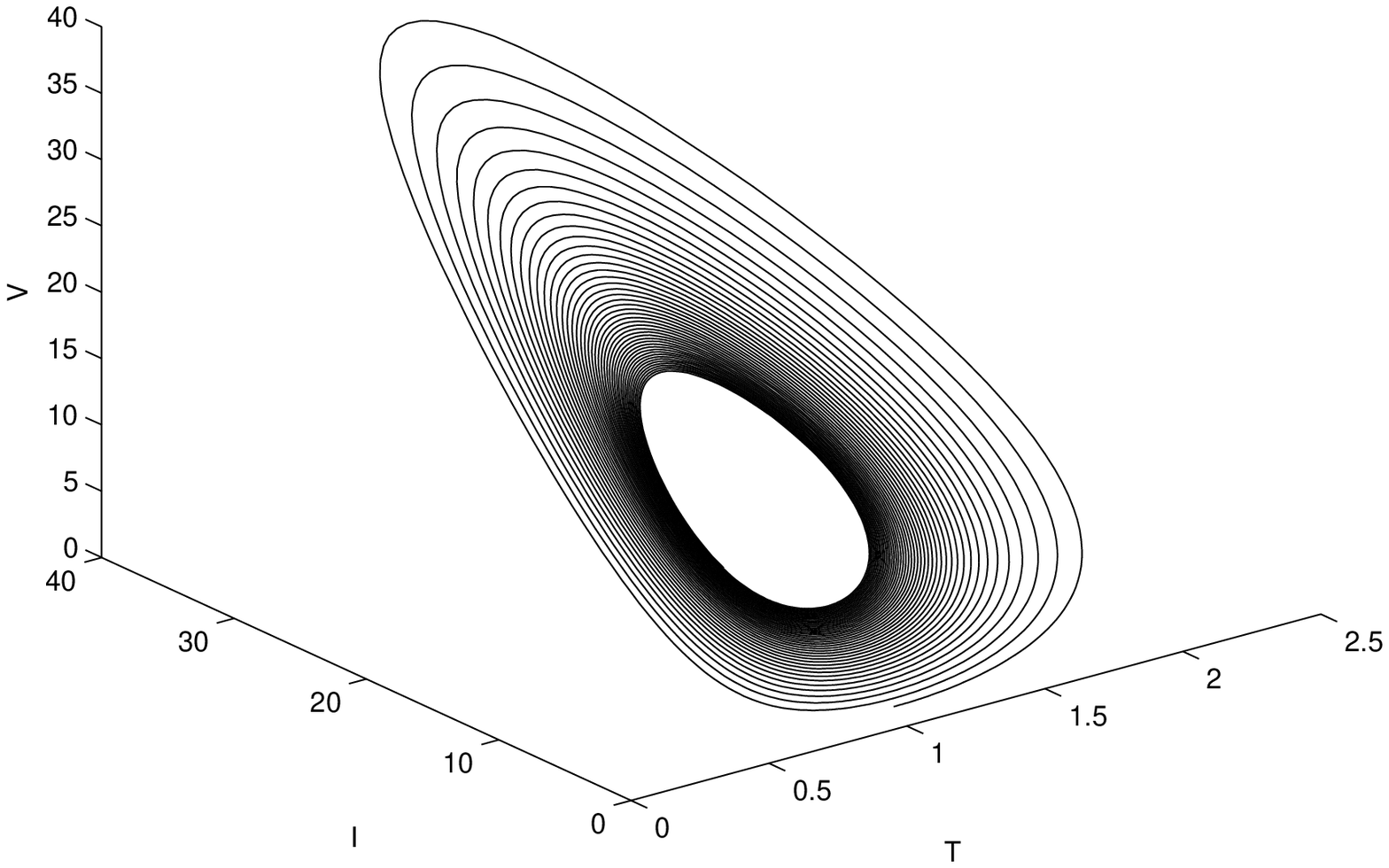}}
\hspace{-0.25in}
\subfigure[Uninfected T-cell population]{\includegraphics[height = 5cm,width =6.75cm]{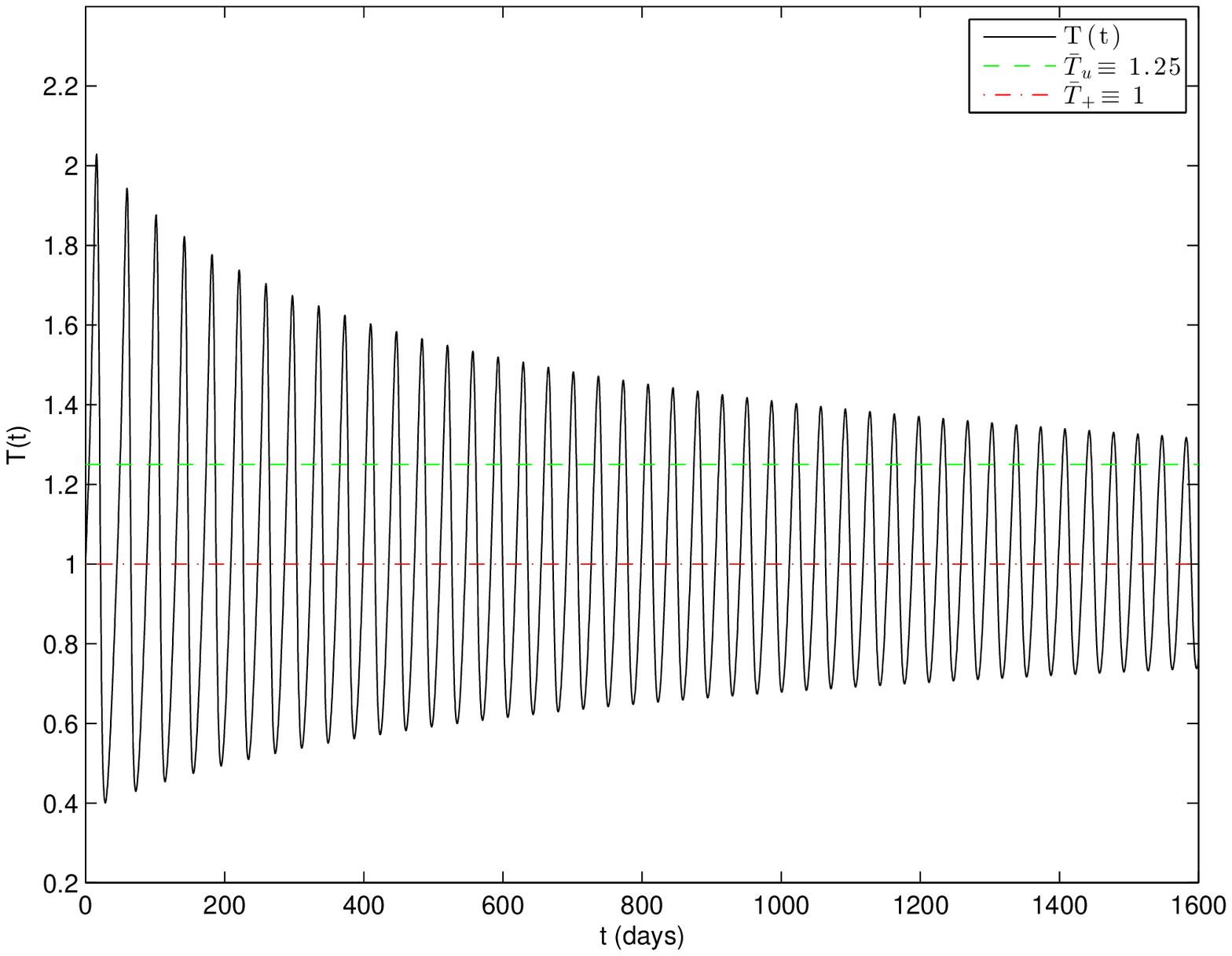}}\\
\hspace{-0.2in}
\subfigure[Virus population]{\includegraphics[height = 5cm,width =6.75cm]{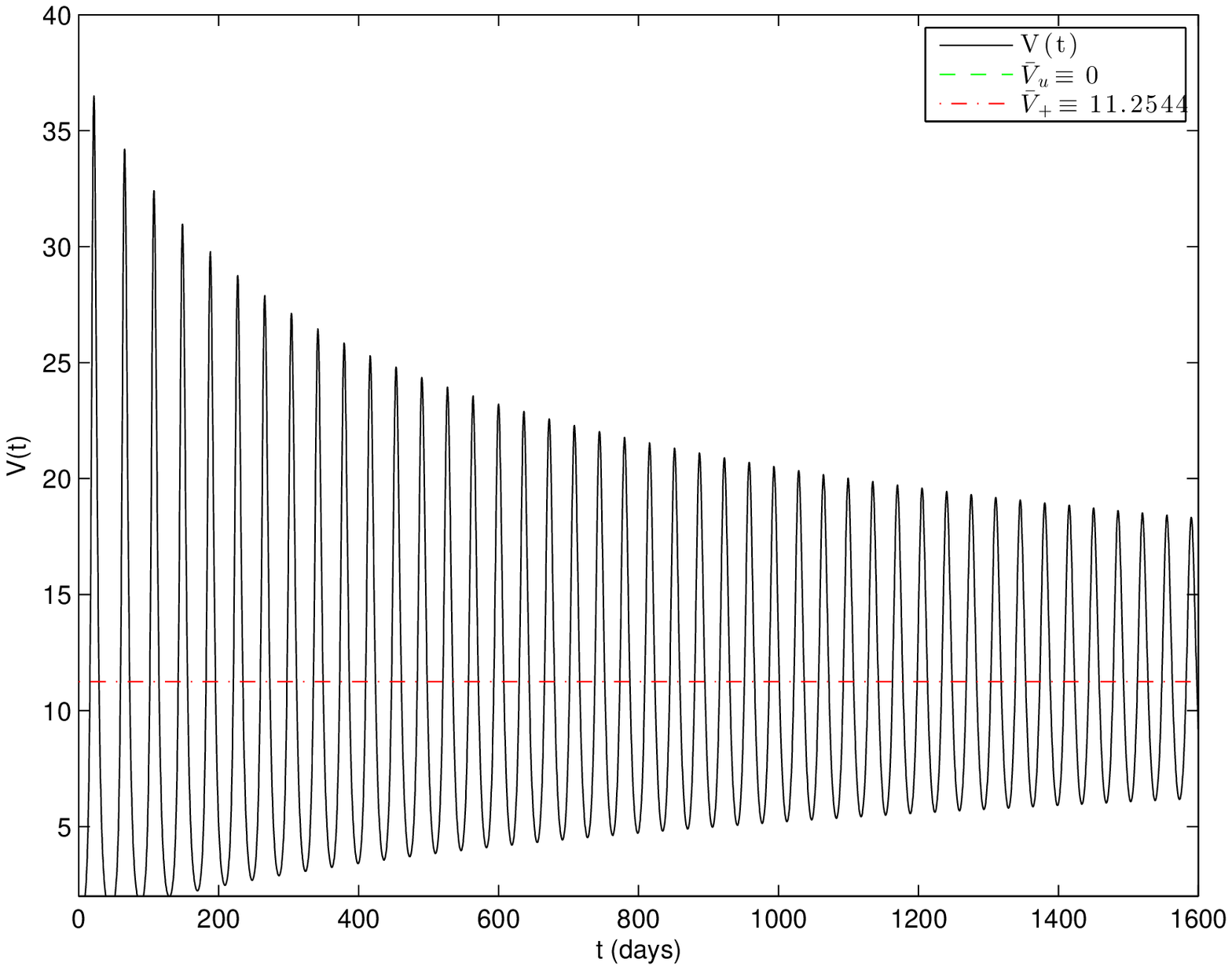}}
\hspace{-0.25in}
\subfigure[Infected T-cell population]{\includegraphics[height = 5cm,width =6.75cm]{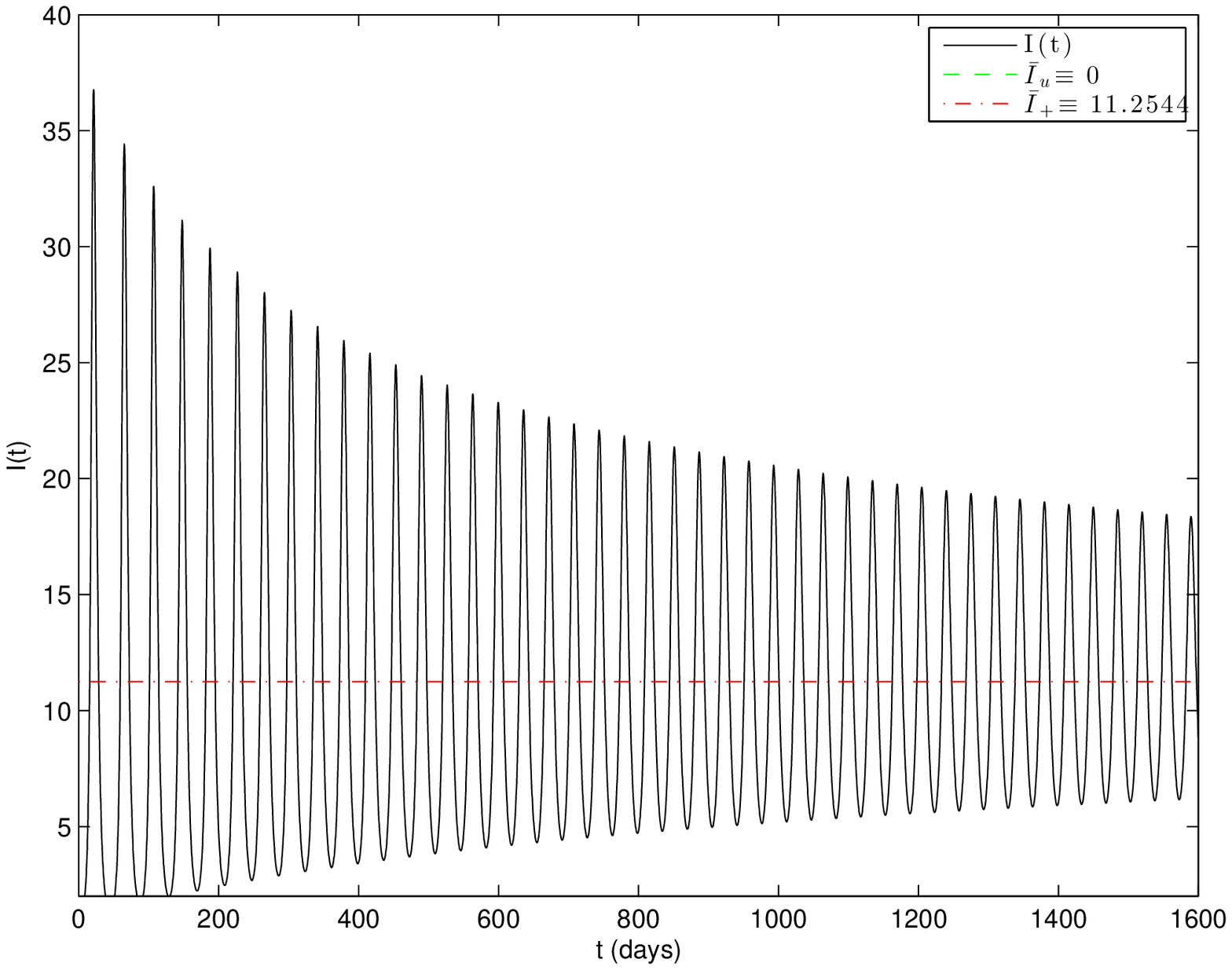}}
\caption{\footnotesize{Simulation of the dimensionless system for $R_m = 9$ and $R_0 = 1.25$. The parameter location is at the border of the (dark) red and unshaded regions of Figure \ref{las-stable}. Hence, complex eigenvalues of the system linearized about $E_i^+$ are of the form $\pm i \beta$, signaling a transition in the asymptotic behavior of solutions and the emergence of a stable orbit.}
}
\label{smallorbit}
\end{figure}

\begin{figure} %
%% COMMENTED FOR SIZE:
\hspace{-0.2in}
\subfigure[TIV phase portrait]{\includegraphics[height = 5cm,width =6.75cm]{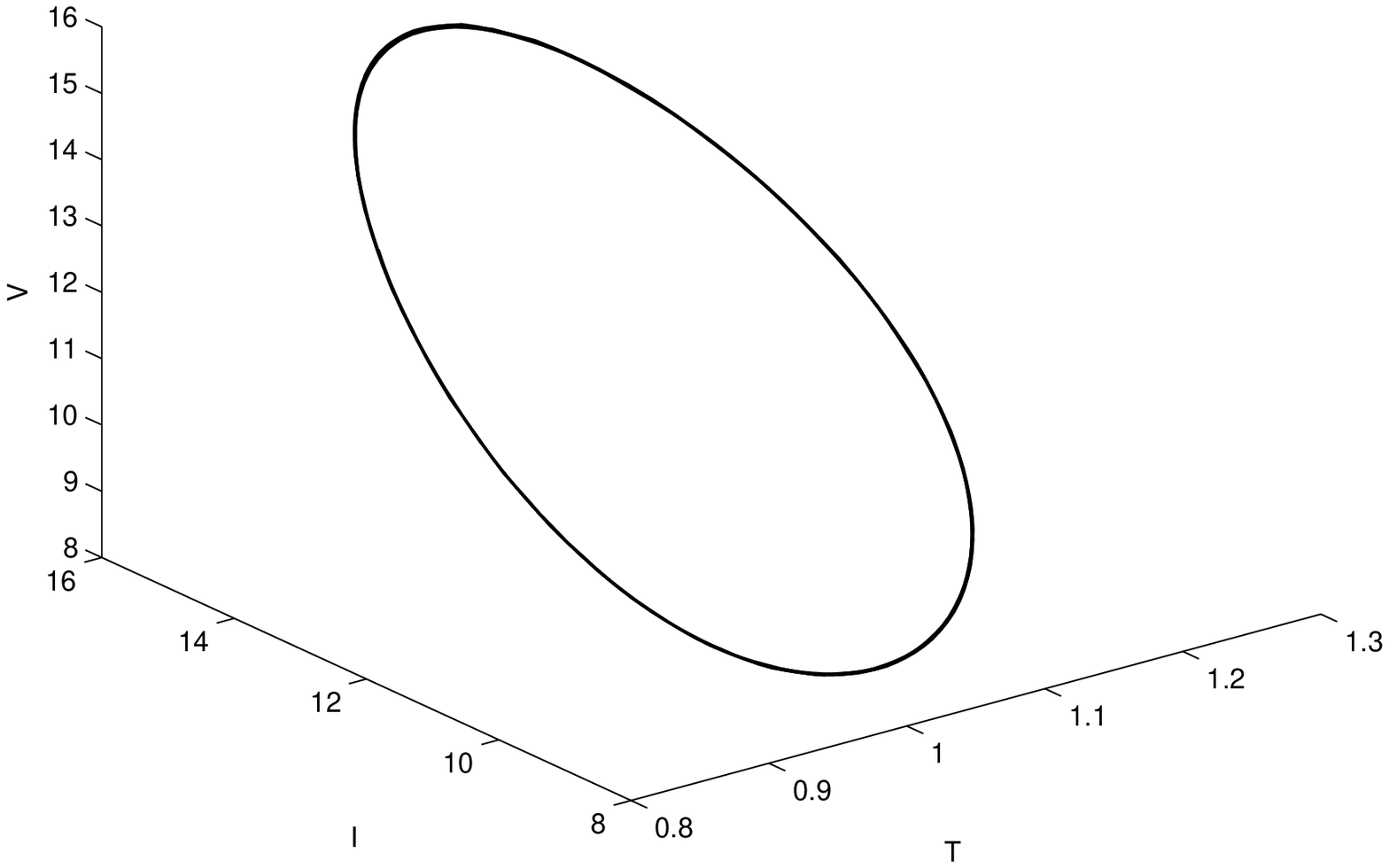}}
\hspace{-0.25in}
\subfigure[Uninfected T-cell population]{\includegraphics[height = 5cm,width =6.75cm]{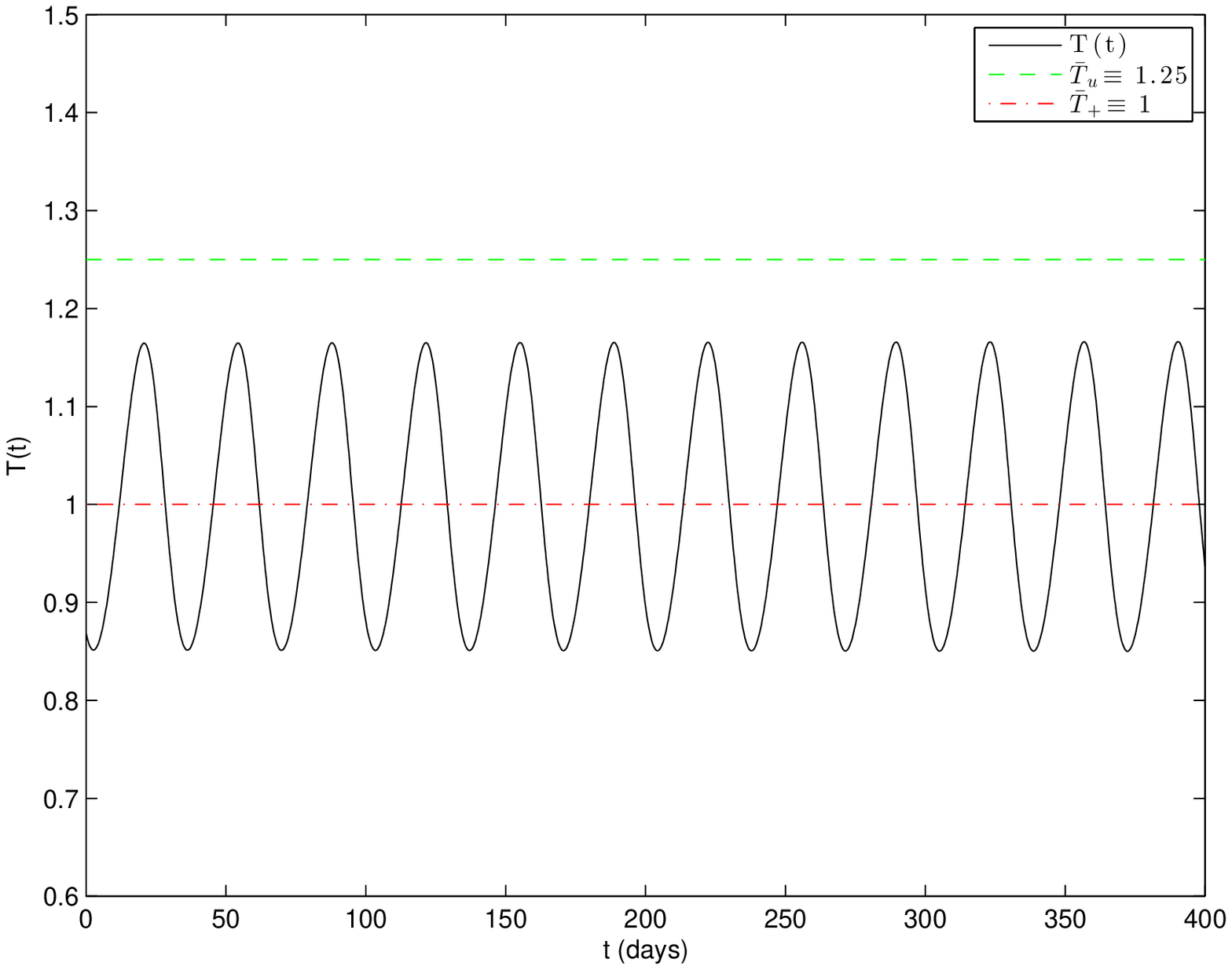}}\\
\hspace{-0.2in}
\subfigure[Virus population]{\includegraphics[height = 5cm,width =6.75cm]{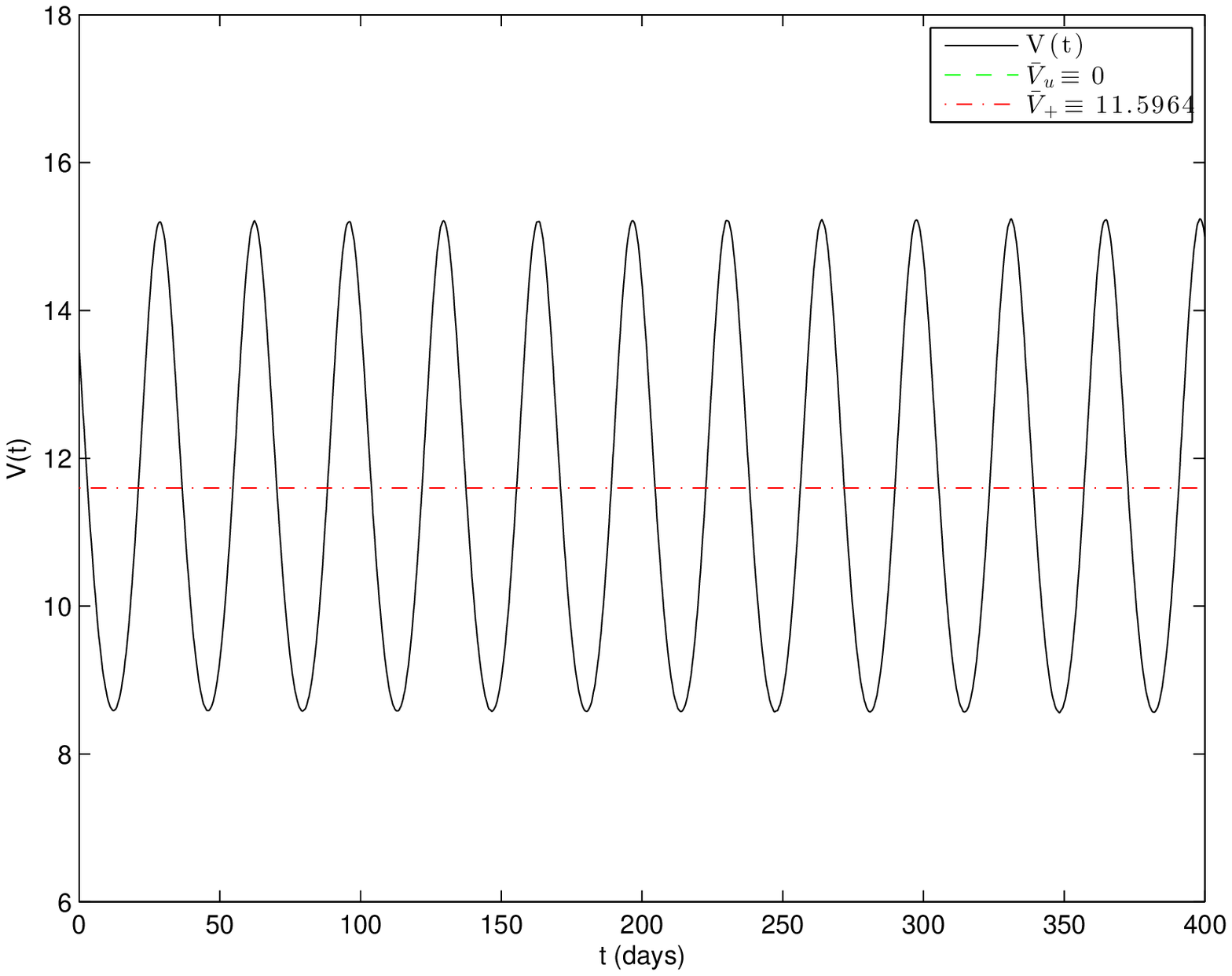}} 
\hspace{-0.25in}
\subfigure[Infected T-cell population]{\includegraphics[height = 5cm,width =6.75cm]{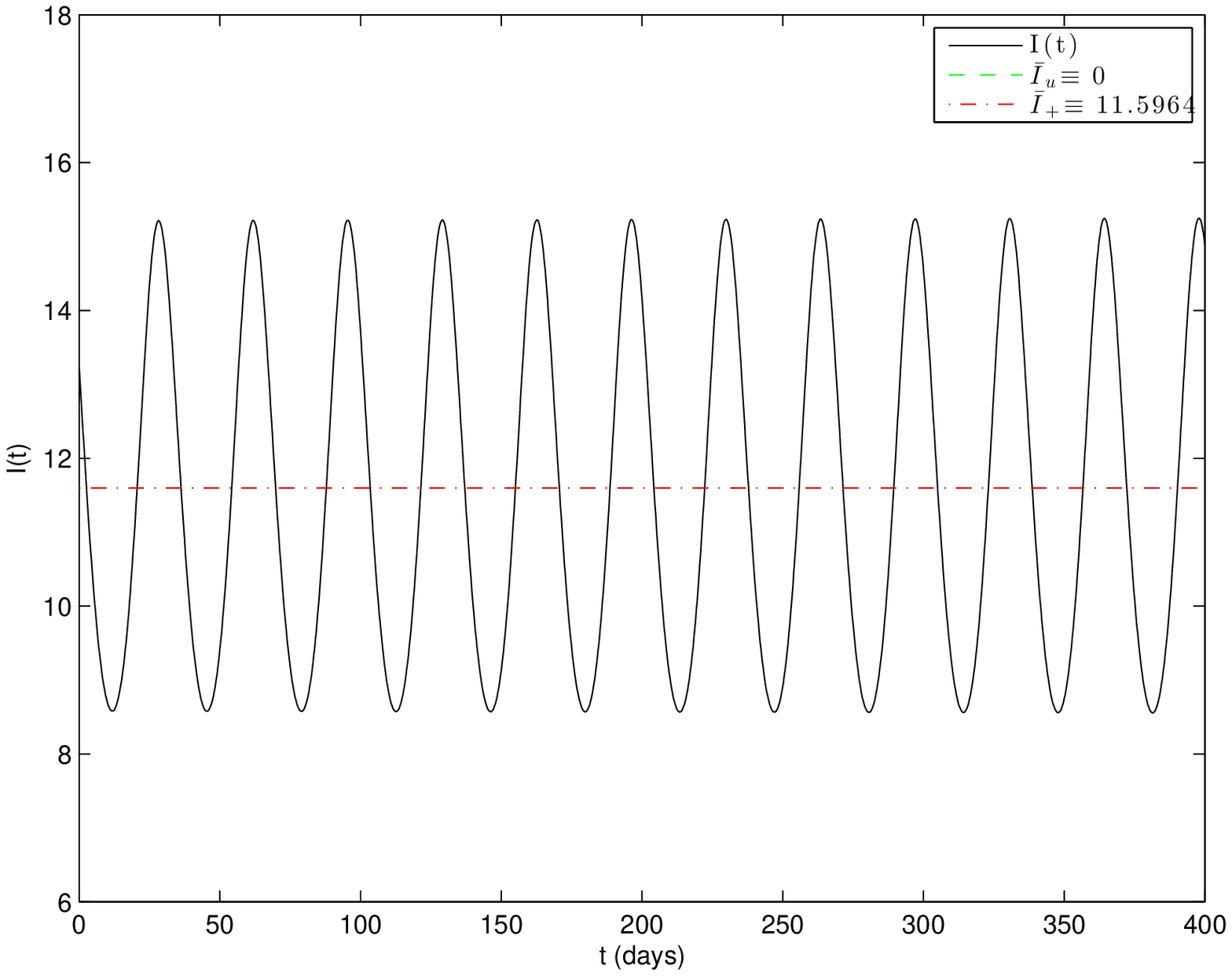}}
\caption{\footnotesize{Simulation of the dimensionless system for $R_m = 9.25$ and $R_0 = 1.25$. The parameters now lie within the unshaded region (Figure \ref{las-stable}). Complex eigenvalues of the system linearized about $E_i^+$ are of the form $\alpha \pm i \beta$ with $\alpha > 0$ yielding instability of the equilibrium. %and, as displayed by the phase portrait, a periodic solution. 
With initial conditions taken near the bifurcation, a stable periodic orbit appears.}
}
\label{largeorbit}
\end{figure}

%\section{Singularly Perturbed Approximation}

\subsection{Conclusions \& Biological Implications}
In the standard model of viral dynamics both the infected and uninfected steady states are globally stable, and thus the stability of equilibria has no dependence on the initial viral load or healthy T-cell count of a host.  The current model features a region of bistability in which both the infection-free and persistent infection steady states are only locally stable.  Therefore, the initial conditions (along with parameter values) influence the long time asymptotic behavior of the system. Though the healthy T-cell count at the time of transmission has the strongest impact on the long term behavior (as shown by Figure \ref{results}), the initial viral load also affects the ability of the virus to establish a persistent infection.
%Hence the model displays that a patient who under goes therapy early on or even prior to infection has a higher probability of clearing infection or maintaining a higher T-cell count throughout infection.  Contrastingly a low T-cell count can lead to a chance for clearance as well as shown by the long time asymptotics in Figure \ref{results}. 
Hence, the proposed model can account for differing infection dynamics that were displayed by clinical experiments due to variations in the initial size of the viral inoculum or initial strength of the T-cell count.  
Additionally, the model yields results that highlight the sensitivity of infection dynamics with respect to variations in initial data.  In particular, Figure \ref{results} shows that only certain ranges of the T-cell count promote viral infection, though their width may change with parameter values. This contradicts the general idea that an increase in CD4 T-cells, which serve as target cells for HIV, will necessarily give rise to a greater likelihood of infection, since within the parameter region of bistability, increasing $T_0$ can actually push the system from the infected state to an uninfected regime.  This result is consistent with the clinical immunosuppressant studies of \cite{Igarashi2, Endo} which found that in some cases, a depleted T-cell count could give rise to a greater possibility of infection. Therefore, differing initial strengths of the susceptible T-cell population may either promote or inhibit viral replication, and a ``sweet spot'' in the initial state (seen as the dark strips within Figure \ref{results}) appears to exist within which infection occurs.
In general, the refined model \eqref{Acute} further highlights the biological implications of T-cell homeostasis, namely that the propensity of the T-cell population to replenish itself, not merely at a constant rate but due to the appearance of a pathogen, can account for differing infection outcomes.
% based upon the initial strength of the viral load and the availability of target T-cells at the time of infection.

With the emergence of a Hopf bifurcation in \eqref{Acute}, there also exists a region of the $(R_m,R_0)$ plane within which solutions display stable oscillations, thereby mimicking observed biological behavior such as viral blips.  Such oscillatory blips are transient spikes in the size of the viral load which often occur in a patient undergoing antiretroviral therapy during the chronic stage of infection \cite{AlanRong}.  Various studies have found that 20\% to 60\% of patients with viral suppression experience viral load blips (depending on the antiretroviral regimen used and the frequency of viral load testing), and perhaps one-third of these experience repeated blips.  Within the proposed model, blips can occur from a sudden, but small, change in parameter values that increase $R_m$ or decrease $R_0$ from the region in which $E_i^+$ is stable into the region in which a periodic orbit becomes stable.
Such changes in parameter values may realistically arise from interruption or sudden alteration of antiretroviral therapy, which strongly influences the values of $k$ and $p$ in the model, and thus the dimensionless parameters $R_0$ and $R_m$, as well.
Hence, \eqref{Acute} may be extended to further explain phenomena during chronic infection when the effects of antiretroviral therapy can drastically alter parameter values in the $(R_m,R_0)$ plane.
Finally, the proposed model can also be generalized to describe the entire time course of infection \cite{Hadji, Hern} by incorporating additional components and biological effects, though a full dynamical analysis would likely be prohibitive in such a case.

%\begin{thebibliography}{1}
%\bibitem{GoMiSa} {\sc M. Goossens, F. Mittelbach, and A. Samarin},
%{\em The} \LaTeX\ {\em Companion}, Addison-Wesley, Reading, MA, 1994.
%
%\bibitem{Higham} {\sc N.~J. Higham}, {\em Handbook of Writing for
%the Mathematical Sciences}, Society for Industrial and Applied
%Mathematics, Philadelphia, PA, 1993.
%
%\bibitem{Lamport} {\sc L. Lamport}, \LaTeX: {\em A Document
%Preparation System}, Addison-Wesley, Reading, MA, 1986.
%
%\bibitem{SerLev} {\sc R. Seroul and S. Levy}, {\em A
%Beginner's Book of} \TeX, Springer-Verlag, Berlin, New
%York, 1991.
%\end{thebibliography}

%\bibliographystyle{plain}
%\bibliography{spatialBib}

\appendix
%\addappheadtotoc
%\appendixpage
\appendixpageoff

\section{Dimensionless System}
\label{appA}
%Non-dimensionalized equations

We begin by placing the model in non-dimensional form.
First, define the dimensionless populations by
$$T^* = \frac{T}{T_c}, \qquad I^* = \frac{I}{I_c}, \qquad V^* = \frac{V}{V_c}$$
where $T_c, I_c$, and $V_c$ are constants to be determined.
Additionally, we scale the time dimension by letting $t^* = \frac{t}{t_c}$. 
Substituting these expressions within \eqref{Acute}, we find
\begin{equation} 
\label{rescale}
\left.
\begin{aligned}
\frac{dT^*}{dt^*} \quad &= \quad \frac{\lambda t_c}{T_c} - \frac{\rho t_c V_c}{C + V_c V^*} T^*V^*- kt_cV_c T^*V^* - d_T t_c T^* \\
\frac{dI^*}{dt^*} \quad &= \quad \frac{kt_c T_c V_c}{I_c} T^*V^* - d_I t_c I^* \\ 
\frac{dV^*}{dt^*} \quad &= \quad \frac{p t_c I_c}{V_c} I^*   - d_V t_c V^*
\end{aligned}
\right \}
\end{equation}
Since a time scale must be selected, we choose $t_c = \frac{1}{d_T}$, though the choices of $\frac{1}{d_I}$ or $\frac{1}{d_V}$ are also reasonable and do not greatly alter the analysis. Next, we choose the population scaling for the dependent variables so as to eliminate parameters in each equation. In particular, we choose
$$T_c =\frac{d_I d_V}{pk}, \qquad I_c =\frac{d_T d_V}{pk}, \qquad V_c = \frac{d_T}{k}.$$
Removing the starred notation from the population variables for convenience, this finally yields the dimensionless system
\begin{equation*} 
\left.
\begin{aligned}
\frac{dT}{dt} \quad &= \quad R_0 + \frac{R_m}{1 + \beta V} TV - TV - T\\
\frac{dI}{dt} \quad &= \quad \alpha_1 (TV - I) \\ 
\frac{dV}{dt} \quad &= \quad \alpha_2 (I - V)
\end{aligned}
\right \}
\end{equation*}
where $R_0, R_m, \alpha_1, \alpha_2, \beta$ are given by \eqref{R0Rm}.
%$$\alpha_1 = \frac{d_I}{d_T}, \qquad \alpha_2 = \frac{d_V}{d_T}, \qquad \beta = \frac{d_T}{Ck} .$$
%
%
\section{Equilibria and Restrictions on Parameter Values}
\label{appB}

Within this appendix, we consolidate results concerning equilibrium states of the model. 
First, we derive all equilibrium solutions and prove Theorem \ref{TSS}.
\begin{proof}(Theorem 3.1)
Beginning with the steady system \eqref{SteadyNd}, we use the third equation to find $I = V$.  Inserting this within the second equation yields
$$I(T-1) = 0.$$
Thus, either $I = 0$ or $T = 1$.  In the former case, $V = 0$, the first equation is exactly $T = R_0$, and the steady state $E_u$ is determined.  In the latter case, the first equation yields a quadratic in $V$, namely 
$$\beta V^2 + \left [ \beta(1-R_0) + 1 - R_m \right ]V + 1- R_0 = 0$$
upon multiplying by $1+\beta V$ throughout.  
This gives rise to two different solutions and with $T = 1$ and $I = V$, these constitute $E_i^+$ and $E_i^-$.
\end{proof}

Next, we derive the restrictions on parameters obtained by enforcing the condition that the computed steady states be real-valued and positive.
First, we note that the condition $b^{2} - 4 a c \geq 0$ imposes the constraint that all equilibria are real-valued, and upon simplification, is equivalent to 
$$D^2 - 2D(1 + R_m) + (1 - R_m)^2 \geq 0$$
where $D = \beta(1-R_0)$.
Simplifying further, this condition becomes
$$ [D - (1 + \sqrt{R_m})^{2}] [D - (1 - \sqrt{R_m})^{2}] \geq 0.$$
Clearly, $(1 + \sqrt{R_m})^2 > (1 - \sqrt{R_m})^2$.  Thus, $b^{2} - 4 a c \geq 0$ if and only if
$$D - (1 -\sqrt{R_m})^{2} \leq 0 \quad \mathrm{or} \quad  D - (1 + \sqrt{R_m})^{2} \geq 0$$
and this simplifies to the statement
$$R_0 \geq 1 - \frac{1}{\beta}(1 - \sqrt{R_m})^{2} \qquad \mathrm{or} \qquad R_0 \leq 1 - \frac{1}{\beta}(1 + \sqrt{R_m})^{2}.$$
Because $\beta < 1$, the region in $(R_m,R_0)$ space described by $R_0 \leq 1 - \frac{1}{\beta}(1 + \sqrt{R_m})^{2}$ contains only negative values of $R_0$, which is not possible for positive original parameter values.  Thus, we only focus on the region in which $R_0 \geq 1 - \frac{1}{\beta}(1 - \sqrt{R_m})^{2}$ to ensure all equilibria possess real values.  

In requiring all infected equilibrium populations to be positive, other restrictions are needed.  We consider two distinct cases.
First, assume $b < 0$.  This inequality is equivalent to $R_0 > 1 + \frac{1}{\beta}(1 - R_m)$.  Hence, for $E_i^{+}$, we find $-b + \sqrt{b^{2} - 4 a c} \geq 0$ so all populations of this infected steady state are positive.  Contrastingly, for $E_i^{-}$, we must have $\sqrt{b^{2} - 4 a c} < \vert b \vert$ to guarantee positive equilibria.   Thus for $E_i^{-}$,  an added requirement is necessary, namely $ac > 0$, which is exactly the condition $R_0 < 1$.
	
Next, assume $b \geq 0$.  Then, this is equivalent to $R_0 \leq 1 + \frac{1}{\beta}(1 - R_m)$. For $E_i^{+}$ we must further consider two subcases, namely $ac \geq 0$ and $ac < 0$.  
If $ac \geq 0$, then $-b + \sqrt{b^{2} - 4 a c} \leq 0$ and $E_i^{+}$ possesses either negative or vanishing infected T-cell and virus equilibrium populations.  Instead, if $ac < 0$, which is equivalent to $R_0 > 1$, then  $-b + \sqrt{b^{2} - 4 a c} > 0$ and $E_i^{+}$ has only positive equilibria.
Contrastingly, considering $E_i^{-}$, both the infected T-cell and virus population are nonpositive, thus the steady state does not exist in this region.
See Figure \ref{exist} for a graphical summary of these restrictions.  

To end Appendix \ref{appB}, we prove Lemma \ref{d0pos}.
\begin{proof}(Lemma \ref{d0pos})
First, note that if $R_0 > 1$, then the inequality is trivially satisfied since $\ol{V}_+ > 0$.  
Hence, we need only consider the case \eqref{Eplus1} within the region of existence.
From \eqref{Vbar}, we see that $\ol{V}_+$ satisfies
$$a \ol{V}_+^2 + b \ol{V}_+ + c = 0$$
where $a = \beta$, $b =  \beta(1-R_0) + 1 - R_m$, and $c = 1 - R_0$.
In this notation, \eqref{Eplus1} is exactly $b < 0$.
Of course, in the region of existence of $E_i^+$, we also have $\ol{V}_+ > 0$.
Additionally, notice that the inequailty $\beta \ol{V}_+^2 + R_0 - 1 > 0$ is equivalent to $a \ol{V}_+^2 - c > 0$, or using the above quadratic, $-b\ol{V}_+ - 2c > 0$.
Since $b < 0$, this is further equivalent to the inequality
$$\ol{V}_+ > \frac{2c}{-b},$$
and we will focus on proving this condition.

Now, if $c \leq 0$, then 
$$\ol{V}_+ > 0 \geq \frac{2c}{-b}$$ 
and the condition is satisfied.
Alternatively, if $c >0$ then writing the root of interest, namely
$\overline{V}_+ = \frac{-b + \sqrt{b^2 - 4ac}}{2a},$
and multiplying by the conjugate we find
$$\overline{V} = \frac{2c}{-b - \sqrt{b^2 - 4ac}} > \frac{2c}{-b}$$
because $b < 0, c > 0$, and $b^2- 4ac > 0$.
Hence, the condition is satisfied in both cases, and the proof is complete.
\end{proof}

\section{Proof of Stability Theorem}
\label{appC}

Finally, we prove Theorem \ref{T1}.
\begin{proof}(Theorem \ref{T1})
We first define $x = (T,I,V)^T$, label the right side of the system \eqref{AcuteNd} by $f$, so that
$$ f(x) = \left(
	\def\arraystretch{2}\begin{array}{c}
	R_0 + \left (\frac{R_m}{1 + \beta V} -1\right )TV - T \\ 
	 \alpha_1 \left ( T V - I \right )\\
	\alpha_2 (I- V) 
\end{array} \right)$$
and compute the gradient of this function
\begin{equation*}
\begin{large}
\centering
	\nabla f(x) = \left(
	\def\arraystretch{2}\begin{array}{ccc}
		-\frac{1}{1 + \beta V} \left [ \beta V^2 + (1 + \beta - R_m) V + 1\right ] & 0 & \left [ \frac{R_m}{(1 + \beta V)^2} - 1 \right ] T \\
		\alpha_1 V & -\alpha_1 &  \alpha_1 T\\
		0 & \alpha_2 & -\alpha_2 
\end{array} \right).
\end{large}
\end{equation*}

To prove the stability result concerning the uninfected steady state, we evaluate the Jacobian at $E_u$ to find
\begin{equation*}
\begin{large}
\centering
	\nabla f(E_u) = \left(
	\def\arraystretch{2}\begin{array}{ccc}
		-1 & 0 & (R_m -1) R_0 \\
		0 & -\alpha_1 &  \alpha_1 R_0\\
		0 & \alpha_2 & -\alpha_2 
\end{array} \right).
\end{large}
\end{equation*}
Clearly, the eigenvalues are $\eta_1 = -1$ and $\eta_2$, $\eta_3$ given by the two roots of the quadratic
$$\eta^2 + (\alpha_1 + \alpha_2) \eta + \alpha_1 \alpha_2 (1 - R_0) = 0.$$
By the Routh-Hurwitz criterion, the latter eigenvalues both have negative real part if and only if $R_0 < 1$.  Hence, we conclude by the Hartman-Grobman theorem that $E_u$ is locally asymptotically stable if $R_0 < 1$.  Contrastingly, if $R_0 > 1$ then $E_u$ is unstable.

Next, we establish the stability properties of the infected steady states.  Again, we compute the Jacobian, but evaluate it only using the T-cell steady state value, $T = 1$, as the value of the viral load differs for $E_i^{+}$ and $E_i^{-}$.
With this, we find
\begin{equation*}
\begin{large}
\centering
	\nabla f(E^{\pm}_i) = \left(
	\def\arraystretch{2}\begin{array}{ccc}
		-R_0 & 0 & \frac{R_m}{(1+ \beta \overline{V}_\pm)^2} - 1 \\
		\alpha_1 \overline{V}_\pm & -\alpha_1 &  \alpha_1\\
		0 & \alpha_2 & -\alpha_2 
\end{array} \right)
\end{large}
\end{equation*}
where
$\overline{V}_\pm$ is given by Theorem 3.1 for either steady state and, in both cases, satisfies the quadratic
\begin{equation}
\label{Veqn}
\beta \overline{V}_\pm^2 + \left [ \beta(1-R_0) + 1 - R_m \right ]\overline{V}_\pm + 1 - R_0= 0.
\end{equation}
The associated characteristic polynomial is $\eta^{3} + d_2 \eta^{2} + d_1 \eta + d_0 = 0$ where
\begin{align*}
d_0 &= \frac{\alpha_1 \alpha_2 \overline{V}_\pm}{(1 + \beta \overline{V}_\pm)^2} \left [-R_m + (1+ \beta \overline{V}_\pm)^2 \right ],\\
d_1 &= (\alpha_1 + \alpha_2)R_0, \\
d_2 &= \alpha_1 + \alpha_2 + R_0.
\end{align*}
The Routh-Hurwitz conditions are clearly met for $d_1$ and $d_2$ since all parameter values are positive.  We only concern ourselves with the sign of $d_0$ and showing the other remaining condition, namely $D_2 = d_1 d_2 - d_0 > 0$.  
After some algebra and use of \eqref{Veqn}, we can rewrite $d_0$ as
$$d_0 = \frac{\alpha_1 \alpha_2}{1 + \beta \overline{V}_\pm} \left [\beta \overline{V}_\pm^2 + R_0 - 1 \right ]$$
which means that $d_0 > 0$ if and only if the condition $$\frac{\alpha_1 \alpha_2}{1 + \beta \overline{V}_\pm} \left [\beta \overline{V}_\pm^2 + R_0 - 1 \right ] > 0$$
is satisfied.  First considering the $E_i^+$ steady state, we find by Lemma \ref{d0pos} that the above inequality holds at every point in the region of existence.
To study the condition $D_2 > 0$, we merely note that this is equivalent to
$$ \frac{\alpha_1 \alpha_2}{1 + \beta\overline{V}_+} \left [ \beta \overline{V}_+^2 + R_0 - 1 \right ] <  (\alpha_1 + \alpha_2)R_0 (\alpha_1 + \alpha_2 + R_0)$$ and results in \eqref{stabilitycondition}.
Thus, the condition for stability of $E_i^+$ is complete.

Finally, to establish the instability of $E_i^-$ in every parameter regime that guarantees postivity of this equilibrium, we will show that the conditions which ensure $\overline{V}_- > 0$ in this case violate the stability criteria.  It was previously shown that the conditions $b < 0$, $c > 0$, and $b^2 - 4ac \geq 0$ are needed in order to arrive at a positive root $\overline{V}_-$ for $E_i^-$ which satisfies \eqref{Veqn}, in which case
$$a = \beta, \qquad b = \beta(1-R_0) + 1- R_m, \qquad c = 1- R_0.$$
Hence, writing the root of interest, namely
$$\overline{V}_- = \frac{-b - \sqrt{b^2 - 4ac}}{2a},$$
and multiplying by the conjugate of the numerator, we find
$$\overline{V}_- = \frac{2c}{-b + \sqrt{b^2 - 4ac}} \leq \frac{2c}{-b}$$
as $b < 0, c > 0$, and $b^2- 4ac \geq 0$.
Rewriting this inequality as $b\overline{V}_- + c \geq - c$ and substituting parameters for $b$ and $c$ yields
$$ [\beta(1-R_0) + 1- R_m] \overline{V}_- + 1- R_0 \geq - (1-R_0).$$
Using \eqref{Veqn}, this inequality is exactly
$$ - \beta\overline{V}_-^2 \geq - (1- R_0).$$
Rearranging finally yields the condition
$$ \beta \overline{V}_-^2 + R_0 - 1 \leq 0,$$
which, considering the postivity of $\overline{V}_-$ and parameters, violates the stability criterion
$$\frac{\alpha_1 \alpha_2}{1 + \beta \overline{V}_\pm} \left [\beta \overline{V}_\pm^2 + R_0 - 1 \right ] > 0,$$
thereby implying instability of $E_i^-$.
\end{proof}

%\begin{figure}[t] %old?
%\centering
%\subfigure[Stability region of $E_i^{+}$]{\includegraphics[height=.45\textwidth]{stablitygraphEiplus.eps}}
%\hspace{0.1in}
%\subfigure[Stability region of $E_i^{-}$]{\includegraphics[height=.45\textwidth]{stablitygraphEiminus.eps}}
%\caption{Regions of the $(R_m,R_0)$-plane where the Ruth-Hurwitz criteria are satisfied.  The blue region denotes where $ d_1 d_2 > d_0$ and the yellow region indicates where $d_0 > 0$.  The tan region is where both inequalities are satisfied and the steady states are stable.}
%\label{stable}
%\end{figure}

\end{document}